\documentclass[prl,showpacs,preprintnumbers,amsmath,amssymb]{revtex4-1}
\usepackage{amsmath}
\usepackage{amsthm}
\usepackage{amssymb}
\usepackage{graphics}
\usepackage{bbold}
\usepackage{graphicx}
\usepackage{pgf}
\usepackage{setspace}
\usepackage{amsfonts}
\usepackage{color}
\usepackage{mathtools}
\usepackage{enumerate}
\usepackage{appendix}
\usepackage{bbold}

\newcommand{\defeq}{\vcentcolon=}

\newcommand{\bra}[1]{\langle #1|}
\newcommand{\ket}[1]{|#1\rangle}

\newcommand{\intfty}{\ensuremath{\int _{-\infty} ^{\infty}} }
\newcommand{\Om}{\ensuremath{\Omega}}

\newcommand{\Lap}{\ensuremath{\Delta }}
\newcommand{\La}{\ensuremath{\Lambda }}
\newcommand{\wt}{\ensuremath{\widetilde}}
\newcommand{\pat}{\ensuremath{\partial}}
\newcommand{\be}{\begin{equation}}
\newcommand{\ee}{\end{equation}}
\newcommand{\ba}{\begin{eqnarray}}
\newcommand{\ea}{\end{eqnarray}}
\newcommand{\bi}{\begin{itemize}}
\newcommand{\ei}{\end{itemize}}
\newcommand{\bn}{\begin{enumerate}}
\newcommand{\en}{\end{enumerate}}
\newcommand{\bp}{\begin{proof}}
\newcommand{\ep}{\end{proof}}
\newcommand{\bt}{\ensuremath{\textbf}}
\newcommand{\mbf}{\ensuremath{\mathbf}}
\renewcommand{\bm}{\ensuremath{\mathbb}}
\newcommand{\mc}{\ensuremath{\mathcal}}
\newcommand{\mr}{\ensuremath{\mathrm}}

\newcommand{\dal}{\ensuremath{\Box }}
\newcommand{\eps}{\ensuremath{\epsilon }}

\newcommand{\la}{\ensuremath{\lambda}}

\newcommand{\ip}[2]{\ensuremath{\langle {#1} , {#2} \rangle}}
\newcommand{\dom}[1]{\ensuremath{\mathrm{Dom} ({#1}) }}
\newcommand{\ran}[1]{\ensuremath{\mathrm{Ran} ({#1}) }}

\newtheorem{thm}{Theorem}

\newtheorem{lemming}{Lemma}

\begin{document}

\bibliographystyle{unsrt}

\title[Covariant UV cutoff for fields on spacetime]{A fully covariant information-theoretic ultraviolet cutoff for scalar fields in expanding FRW spacetimes}

\author{A. Kempf, A. Chatwin-Davies}
\address{ Department of Applied Mathematics, University of Waterloo}
\author{R.T.W Martin}
\address{Department of Mathematics and Applied Mathematics, University of Cape Town}

\begin{abstract}

While a natural ultraviolet cutoff, presumably at the Planck length, is widely assumed to exist in nature, it has proven difficult to implement a minimum length scale covariantly. A key reason is that the presence of a fixed minimum length would seem to contradict the ability of Lorentz transformations to contract lengths. In this paper, we implement a fully covariant Planck scale cutoff by cutting off the spectrum of the d'Alembertian. In this scenario, consistent with Lorentz contractions, wavelengths that are arbitrarily smaller than the Planck length continue to exist. However, the dynamics of modes of wavelengths that are significantly smaller than the Planck length possess a very small bandwidth. This has the effect of freezing the dynamics of such modes. While both, wavelengths and bandwidths, are frame dependent, Lorentz contraction and time dilation conspire to make the freezing of modes of transplanckian wavelengths covariant. In particular, we show that this ultraviolet cutoff can be implemented covariantly also in curved spacetimes. We focus on Friedmann Robertson Walker (FRW) spacetimes and their much-discussed transplanckian question: The physical wavelength of each comoving mode was smaller than the Planck scale at sufficiently early times. What was the mode's dynamics then? Here, we show that in the presence of the covariant UV cutoff, the dynamical bandwidth of a comoving mode is essentially zero up until its physical wavelength starts exceeding the Planck length. In particular, we show that under general assumptions, the number of dynamical degrees of freedom of each comoving mode all the way up to some arbitrary finite time is actually finite. Our results also open the way to calculating the impact of this natural UV cutoff on inflationary predictions for the CMB.

\end{abstract}

\maketitle

\section{Introduction}

Heuristic
quantum gravity arguments suggest the existence of some form of a smallest length in nature. Assume, for example, one tried to resolve
a distance with an uncertainty of less than a Planck length. Due to the uncertainty principle, this should imply a
 momentum uncertainty above the Planck momentum. This in turn, due to Einstein's equation, should lead to an uncertainty in  curvature of the order of a Planckian size curvature radius. This curvature uncertainty foils the attempt to measure a distance with a precision below the Planck length.
The difficulty with this intuition lies in expressing the existence of a minimum length covariantly. It would appear that the existence of a fixed minimum length in nature contradicts the ability of  Lorentz transformations to contract any given length.

In this context, it was proposed in \cite{Kempfsamp2,Kempf-cutoff} that the notion of a smallest length in
nature could manifest itself in that all physical fields possess no wavelength smaller than a fixed smallest wavelength. While the notion of a finite minimum wavelength is not covariant either, it can be turned into a covariant notion, as we will discuss. Before, however, it is instructive to collect properties of fields which possess no wavelength smaller than some fixed finite smallest wavelength, say at the Planck length.

To this end, we draw results from information theory. In information theory, the crucial equivalence between continuous and discrete information is established by Shannon sampling theory. It holds that bandlimited functions
 can be perfectly reconstructed anywhere from the values they take
 on any chosen discrete sets of points whose average spacing is given by the inverse of the bandwidth. For example, the amplitudes of 20KHz bandlimited continuous music signals are usually sampled and recorded at the rate corresponding to this bandwidth. Then, using Shannon sampling theory, the continuous music signal can be reconstructed perfectly for all times from only these discrete samples. Inaccuracies can only arise from errors made in taking the samples. While the samples need not be taken equidistantly, in engineering practice they usually are. The reason is that the more unevenly spaced the samples are taken, the more the reconstruction is sensitive to measurement inaccuracies.

If physical fields are bandlimited they too can be reconstructed everywhere from their amplitudes on any discrete lattice of sufficiently dense average spacing. It is worth pointing out  a key difference to physical theories that are defined on a single lattice. In contrast to those theories, here the fields can be represented on continuous spacetime or, equivalently, on \it any \rm discrete lattice whose average spacing is dense enough. The fact that no lattice is singled out implies, for example, that external symmetries such as translation and rotation invariance are not broken. With a minimum wavelength cutoff, spacetime could be in effect simultaneously continuous and discrete in mathematically the same way that information can be simultaneously continuous and discrete, see \cite{Kempf-cutoff-NJP}.

Concretely, functions with a minimum wavelength are called $\Om$-bandlimited, where $\Om>0$ is the
 magnitude of the highest frequency component they contain. On the real line, $\bm{R}$,
 the value that an $\Om-$bandlimited function
 $f$ takes at a point $x\in \mathbb{R}$ is completely reconstructible for example from
the values it takes on the discrete set of points $\{ x_n :=
\frac{n\pi}{\Om} \} _{n \in \mathbb{Z} }$, namely through the Shannon
sampling formula \cite{Shannon}: \be   f(x) = \sum _{n\in
\mathbb{Z}} f(x_n) \frac{\sin \left( \Om (x-x_n) \right)}{ \Om
(x-x_n)} \label{eq:shannon} .\ee

In $\bm{R} ^n$, a bandlimit, or ultraviolet cutoff of this kind can be established
as a cutoff on the spectrum of the Laplacian, $\Lap := - \sum_{j=1}^{n} \frac{d^2}{d(x^j)^2}$. This is because the
subspace of fields whose Fourier transforms have support only on
the disc of radius $\Om$ in $\mathbb{R} ^n$ is that subspace
of $L ^2 (\mathbb{R} ^n)$ spanned by the eigenfunctions of $\Lap$
whose eigenvalues are less than or equal to $\Om ^2 $ in
magnitude. More precisely, this subspace is the range of
$\chi _{[0, \Om ^2 ]} (\Lap )$ where $\chi _{[0 , \Om ^2 ]}$ is
the characteristic function of the interval $[0, \Om ^2 ]$ and
$\chi _{[0 , \Om ^2 ]} (\Lap )$ is a projection defined by the functional calculus.

    This ultraviolet cutoff on fields in flat space is naturally
generalizable to scalar fields on curved space \cite{Kempf-cutoff}. Given an arbitrary curved space, or Riemannian manifold, we can assume its Laplacian to be self-adjoint. If it possesses boundaries, we assume suitable boundary conditions have been chosen. We now restrict the space of
physical scalar fields to be that subspace of square integrable functions on
the manifold spanned by the eigenfunctions of the Laplacian whose eigenvalues are less than or equal to $\Om
^2$, where $\Om$ is the bandlimit or ultraviolet cutoff parameter. That is,
define the space of physical fields on a manifold $M$ to be $B(M,
\Om ) := \chi _{[-\Om ^2 , \Om ^2 ]} (\Lap ) L^2 (M)$, where
$\Lap$ is the Laplacian of the manifold $M$. This
cutoff is covariant since the spectrum of the Laplacian is covariant.
For a proof of the sampling property for certain classes of Riemannian
manifolds, see \it e.g. \rm, \cite{Pesenson}.

In \cite{Kempf-cutoff}, to restore full covariance, a cutoff on the spectrum of the d'Alembertian on Minkowski space has been investigated. This cutoff is covariant because the spectrum of the d'Alembertian is scalar. Furthermore, it indeed provides a way to overcome the  paradox that a fundamental minimum length should not be able to coxist with Lorentz transformations since the latter can contract any length further. Namely, it was found that,
in this scenario, wavelengths that are arbitrarily smaller than the Planck length continue to exist. However, the dynamics of modes of wavelengths that are significantly smaller than the Planck length then automatically possess a very small bandwidth. Their dynamics in effect freezes: It suffices to take samples at a very low rate in time to reconstruct the dynamics at all points in time.
While both wavelengths and bandwidths are frame dependent, Lorentz contraction and time dilation conspire to make this behaviour covariant, as they have to because cutting the d'Alembertian's spectrum is covariant.

Here, we extend this preliminary analysis by showing that this ultraviolet cutoff can be implemented covariantly also in curved spacetimes. We study Friedmann Robertson Walker (FRW) spacetimes in particular, where we focus on the much-discussed trans-planckian question: The physical wavelength of each comoving mode was smaller than the Planck scale at sufficiently early times. What was a mode's dynamics then? (See \cite{Jacobson} for a review of the trans-planckian question, \cite[Section V]{Brandenberger}, \cite{Martin-Brandenberger, Danielsson} for a cosmological introduction, and \cite{shiu05} for a more recent review. Various approaches to the trans-planckian question are studied in \cite{niemeyer01,Kempf-pertspec,brandenberger01,brandenberger05, easther01-1,easther01-2,easther03, greene05,easther05} to name a few.) Here, we will show that in the presence of the covariant UV cutoff on the spectrum of the d'Alembertian, the dynamical bandwidth of a comoving mode is essentially zero up until its physical wavelength starts exceeding the Planck length. In particular, we show that under general assumptions, the number of dynamical degrees of freedom of each comoving mode all the way up to some arbitrary finite time is actually finite. Concretely, the number of samples in time that need to be taken from the beginning up to some finite time is finite. Our results should open the way to calculating also the impact of this natural UV cutoff on inflationary predictions for the CMB.

% \bf [here cite several papers on the transplanckian problem in cosmology - please try to find some that are review-like. For example almost all papers that cite my work with Niemeyer are on the transplanckian problem in cosmology. Notice that we needed to increase the number of our references anyway] \rm

In preparation, let us now recall certain results from the sampling theory of bandlimited functions on $\bm{R} ^n$ that we will later need.

\section{Review of Basic Sampling Theory of bandlimited functions}

    Much is known about the reconstruction and interpolation
properties of bandlimited functions in one dimension. In essence, any bandlimited function can be reconstructed from any
discrete set of points provided that set of points is sufficiently dense
on the real line. Density is defined in the following way. Let
$\La := \{ x_n \} _{n \in \mathbb{Z}} \subset \mathbb{R} $ be a  strictly increasing
sequence of real numbers with a finite minimum spacing between its
members. Such a set $\La$ is called a set of sampling if any
bandlimited function can be reconstructed in a stable fashion from
the values it takes on the points of $\La$. More precisely, $\La$ is called a set of sampling if
there are constants $c, C >0$ such that for every $f \in B(\mathbb{R},\Om)$,
\be c \| f \| ^2 \leq \sum _k |f (x_k ) | ^2 \leq C \| f \| ^2 .\ee  The upper inequality always holds
provided the sampling sequence $\La$ has a finite minimum spacing \cite{Young}.  The double inequality shows that the linear map
from $B(\mathbb{R}, \Om) $ to $l^2 (\bm{Z} )$ which maps $f \in B(\mathbb{R}, \Om)$ onto its sequence of sample values $\{f (x_n) \} $
is invertible.  In particular, this implies that $f$ can be perfectly reconstructed from its sample values and that a bounded
error in the sample values leads to a bounded error in the reconstructed function, \emph{i.e.}, $f$ can be reconstructed
from its sample values in a stable fashion. Now let $n (r)$ be the minimum number of points of $\La$ in any subinterval of length
$r$. The lower Beurling density of a sequence $\La $ is then defined as $D (\La ) := \lim _{r\rightarrow \infty} \frac{n (r)}{r} $. With
this definition, the following theorem holds \cite{Landau1}

\begin{thm}{ (Beurling) }

    The set of points $\La$ is a set of sampling for $B(\mathbb{R}, \Om )$ if $D (\La )
> \frac{\Om}{\pi} $. Conversely, if $\La$ is a set of sampling
then $D (\La ) \geq \frac{\Om}{\pi} $.

\end{thm}

    More generally, one can consider square integrable functions in
$\mathbb{R} ^n$ whose Fourier transforms have support only in some
compact $n-$dimensional set $S$. Call the space of such functions
$B(\mathbb{R}^n,S)$. The necessity part of the previous theorem generalizes to
these frequency limited functions \cite{Landau1}:

\begin{thm}{ (Landau) }
    If the set of points $\La$ is a set of sampling for $B (\mathbb{R}^n,S)$
then $ D (\La ) \geq \frac{ \mu (S) } {(2\pi )^n} $

\label{thm:Landau}
\end{thm}

    Here $\mu (S)$ denotes the Lebesgue measure, or volume of S.
$\La $ is a set of points in $\mathbb{R} ^n $ with a
finite minimum spacing between any two of its members and $n (r)$ is
the smallest number of  points of $\La $ in any $n-$dimensional ball of radius
$r$. This theorem demonstrates that in order for a set of points to be a
set of sampling for frequency limited functions, it must be
sufficiently dense, where the minimum possible density is proportional
to the volume of the compact set $S$ on which the Fourier transforms of the functions in $B(\mathbb{R}^n,S)$
have support.

    Observe that if $f \in B(\mathbb{R}^n,S)$, then there is an $M>0$ such that $S \subset M^n + \mbf{x}$ where
$M^n$ is an $n-$dimensional cube of side length $M$ and $\mbf{x} \in \bm{R} ^n$. It follows, by one-dimensional
sampling in each co-ordinate, that $\La := \{ \bm{Z} ^n \frac{2\pi}{M} \}$ is a set of sampling for
$B(S)$. This shows that sets of sampling always exist for any
space of frequency limited functions. Of course this example of a
set of sampling will have density greater than that required by
Landau's theorem, unless $S$ is itself a cube.

\section{Flat spacetime}

    Let $M$ denote flat $1+d$ dimensional spacetime. Here, the d'Alembertian is simply
$- \frac{\pat ^2}{\pat t ^2 }  + \Delta $, where $\Delta$ is the
spatial Laplacian. The eigenfunctions of the d'Alembertian are
the plane waves $ e^{i(p_0t - \mathbf{p} \cdot \mathbf{x} ) } $,
where $p_0$ is a temporal frequency and $\mathbf{p}$ is the
spatial frequency vector. In order to impose a covariant ultra-violet cutoff $\Omega$, we restrict our attention to the subspace $B(M,\Om)$ of square integrable functions on $M$ spanned by the plane
waves whose frequencies, or eigenvalues, obey the inequality \be
\left| p_0^2 - |\mathbf{p}|^2 \right| \leq \Om ^2.
\label{eq:bandlimit} \ee

\begin{figure}[h!]
\centering
\caption{Fixed spatial $\mbf{p}$ modes have finite temporal bandwidth.}
\includegraphics[scale=.2]{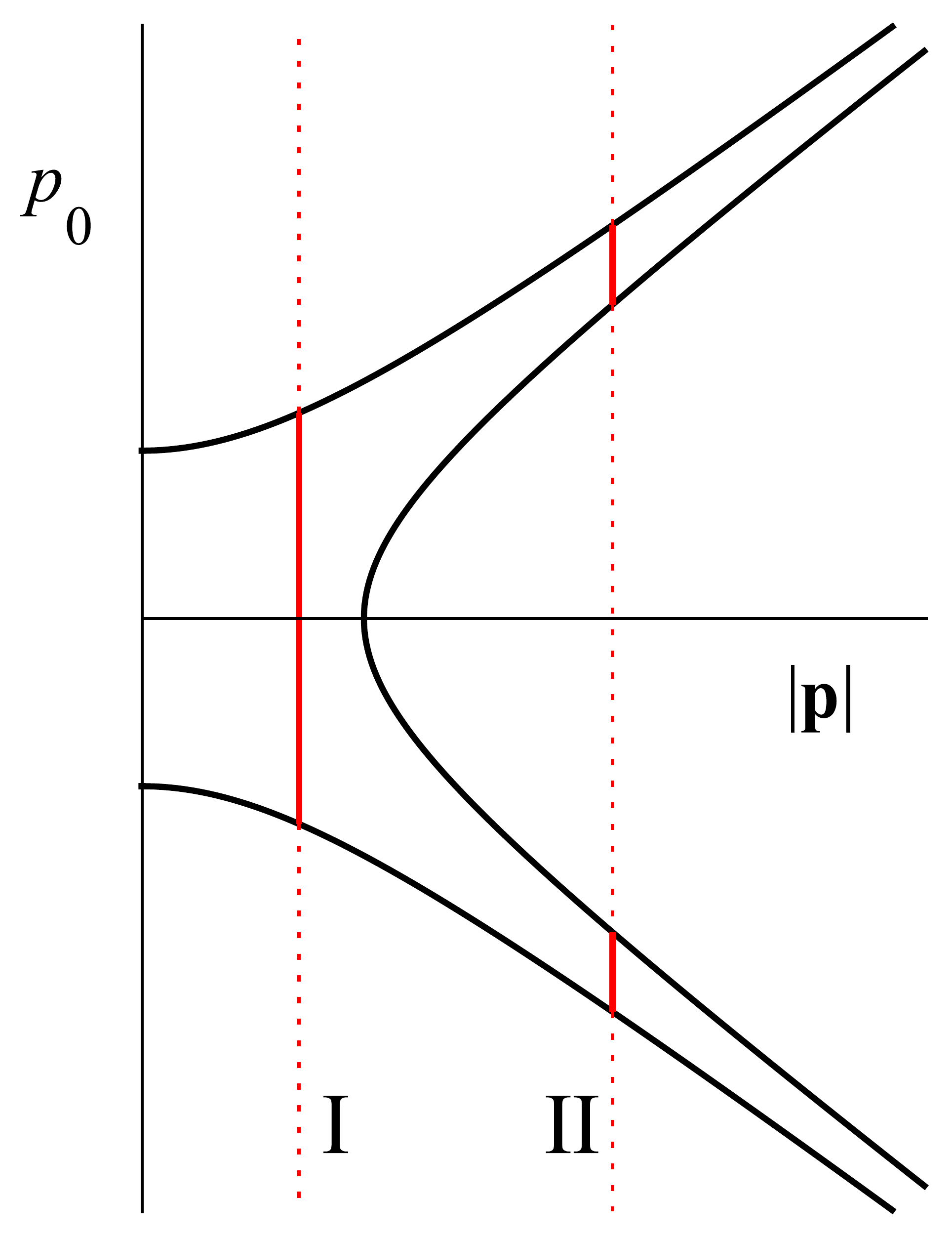}
\label{figure:bandwidth}
\end{figure}

    The set $S$ of all $p_0$ and $\bt{p}$ satisfying this inequality is
the closed interior of a region bounded by hyperboloids. This
is an unbounded region, so there exist values of $p_0$ and
\textbf{p} inside the region which are individually arbitrarily
large in magnitude.  Unlike in the case of bandlimited
functions in flat Euclidean space, there exists no discrete set of
points $\La := \{ x _n \} _{n \in \bm{Z}} \subset M$ of finite
density which is a set of sampling for $B(M, \Om)$. This follows
from the fact that $B(M,\Om)$ is the image of $L^2 (S)$ under the
Fourier transform, where $S$ is the region in $\bm{R} ^4$ which
obeys the inequality (\ref{eq:bandlimit}). It is straightforward
to check that for any choice of $1\leq d \leq 3$ the volume of
$S$ is infinite, so that Landau's theorem implies that the minimum
density a set of points needs to have in order to be a set of
sampling for $B(M, \Om )$ is infinite.

    Nevertheless, elements of $B(M, \Om)$ still have special reconstruction properties. Consider an arbitrary
 spatial mode, \emph{i.e.}, consider a field $\phi(t,\mbf{p})$ of fixed spatial momentum $\mbf{p}$. Then the temporal frequencies $p_0$
of that spatial mode $\phi(t,\mbf{p})$ are confined to the compact set
$S_{\mbf{p}}$ described by the inequality \be r_1 := \Re\left\{\sqrt{
|\mbf{p}| ^2 - \Om ^2}\right\} \leq | p_0 | \leq \sqrt{ |\mbf{p}| ^2 + \Om ^2}
=: r_2 . \label{eq:radii} \ee $S_{\mbf{p}}$ is a single interval
$[-r_2, r_2]$ for $ |\mbf{p}| \leq \Om$ and the union of two intervals
$[-r_2 , -r_1 ] \cup [r_1, r_2]$ for $|\mbf{p}|
> \Om $ (see Figure \ref{figure:bandwidth}). For the set $\La = \{t_n\}$
to be a set of sampling for a fixed spatial mode $\mbf{p}$, by Landau's theorem,
$\La$ must have a density at least as great as a minimum density which
is proportional to the length of $S_\mbf{p}$. It is not hard to check that the length of $S_\mbf{p}$ vanishes
in the limit as $|\mbf{p}| \rightarrow \infty$. Therefore,
each fixed spatial mode is temporally bandlimited and as a
function of time belongs to $B (\mathbb{R},S_{\mbf{p}})$. In this simple case
where $S_{\mbf{p}}$ is composed of at most two intervals, it is
known how to explicitly construct sets of sampling which achieve
the minimum density required by Landau's theorem \cite{Linden}.
The fact that the volume of $S_{\mbf{p}}$ vanishes in the limit as
$|\mbf{p}| \rightarrow \infty$ then means that the density of
temporal degrees of freedom of a fixed spatial mode $\mbf{p}$
vanishes in this limit.

    This is the manifestation of the covariant ultra-violet cutoff in
flat spacetime. Increasingly large spatial modes still exist, but
their density of temporal degrees of freedom decreases to zero; they are
``frozen out." Further notice that this cutoff respects Lorentz
symmetry, as it must since it is fully covariant. Suppose
that one performs a Lorentz transformation to new time and space
co-ordinates $(t', \mbf{x} ' )$. If one had a temporal lattice of
sample points $\La := \{ t_n \}$ in the old co-ordinates which was
dense enough to reconstruct a fixed spatial mode $\mbf{p}$ in the
old co-ordinates, then in the new co-ordinates the time
co-ordinate is dilated, so that the image of the set of points
$\La '= \{t_n' \}$ is less dense. However, under this co-ordinate
change, the fixed spatial mode $\mbf{p}$ becomes a larger fixed
spatial mode $\mbf{p} '$ due to length contraction.
Consequently, the density a set of points in the time co-ordinate $t'$
needs to have in order to be a set of sampling for the fixed
spatial mode $\mbf{p} '$ is lower. In other words, the rate at which one needs
to measure the values of a fixed spatial mode in the new
co-ordinates in order to reconstruct it for all time is slower
than the rate required in the old co-ordinates. This allows the image set $\La '$ of $\La$
under a Lorentz transformation to be both less dense and a set of sampling for the
higher spatial frequency mode $\mbf{p} '$.
A further interesting observation is that since there is an
upper bound for the temporal bandwidth volume of any fixed spatial
mode, it is possible that there could be a set of points in time
that is dense enough to be a set of sampling for \emph{every} fixed
spatial mode. Such sets of points $\La$ do indeed exist and are
not difficult to construct. This means that if one knows the
values that a covariantly bandlimited field takes on each spatial
hypersurface corresponding to the fixed values of time belonging
to $\La$, then the field can be reconstructed everywhere for all space
and time. Such a reconstruction formula can be written down
explicitly. Namely, let $\La := \{ x_n \} \cup \{ y_n \}$, where
$x_n := \frac{n\pi}{\sqrt{2} \Om } $ and $y_n := \frac{(n +
\alpha)\pi}{\sqrt{2} \Om }$. Here, $\alpha \in (0,1)$ is arbitrary.
The volume, or length of the set $S _{\mbf{p}}$ is always less
than or equal to $2\sqrt{2} \Om$ for any fixed $ |\mbf{p} |>0$ so
that by Landau's theorem (Theorem \ref{thm:Landau}), this set is
dense enough to be a set of sampling for every fixed spatial mode
$\mbf{p}$. In fact, $\La$ is a set of sampling for each fixed
spatial mode, and if $\phi _{\mbf{p}} $ is a fixed spatial mode,
the reconstruction formula is given by \be \phi _{\mbf{p}} (t) =
\sum _{n \in \bm{Z}} \phi _\mbf{p} ( x_n ) K (t - x_n ) + \phi _{\mbf{p}} (y_n)
K (y_n -t ), \label{eq:recon} \ee where for $|\mbf{p} |< \Om$, \be
K(t) = \frac{\cos (\sqrt{2} \Om t - \alpha \sqrt{2}\Om /2 ) - \cos
(\alpha \sqrt{2} \Om /2 )} {\Om t \sin (\alpha \sqrt{2} \Om /2) },
\ee while for $|\mbf{p} | > \Om$, \small \ba K(t) & = & \frac{\cos
\left( m_{\mbf{p}}  \alpha \sqrt{2} \Om - (\sqrt{2} \Om +
\sqrt{\mbf{p} ^2 - \Om ^2})t \right)} { \sqrt{2} \Om t \sin \left(
m_\mbf{p}  \alpha \sqrt{2}\Om   \right)} \nonumber \\ & -&  \frac{
\cos \left( m_\mbf{p} \alpha \sqrt{2}{\Om} - ( (2m_\mbf{p} -1)
\sqrt{2} \Om  - \sqrt{\mbf{p} ^2 - \Om ^2})t \right) } { \sqrt{2}
\Om t \sin \left( m_\mbf{p}  \alpha \sqrt{2}\Om   \right)}
\nonumber \\ &+ & \frac{\cos \left( (2 m_{\mbf{p}} -1) \alpha
\sqrt{2} \Om /2 - ((2 m_\mbf{p} -1 ) \sqrt{2} \Om - \sqrt{\mbf{p}
^2 - \Om ^2})t \right)} { \sqrt{2} \Om t \sin \left( (2 m_\mbf{p}
-1 ) \alpha \sqrt{2}\Om /2  \right)}\nonumber \\ & -&   \frac{
\cos \left( (2 m_\mbf{p} -1 ) \alpha \sqrt{2}{\Om} /2 -
\sqrt{\mbf{p} ^2 - \Om ^2}t \right) }{ \sqrt{2} \Om t \sin \left(
(2 m_\mbf{p}  -1 ) \alpha \sqrt{2}\Om /2  \right)} \ea \normalsize
\cite{Linden}. In the above, $m_\mbf{p}$ is the largest integer
for which $(m_\mbf{p} -1) \sqrt{2} \Om < \sqrt{\mbf{p} ^2 - \Om
^2}$. It follows that if one is given a covariantly bandlimited
field $\varphi (t , \mbf{x} )$ and if one knows the values $\{
\varphi ( t_n , \mbf{x} ) \} _{t_n \in \La ; \mbf{x} \in \bm{R} ^
3 } $, then this field can be reconstructed perfectly everywhere,
for all space and time. This is because the knowledge of the values $\{ \varphi ( t_n ,
\mbf{x} ) \} _{t_n \in \La ; \mbf{x} \in \bm{R} ^ 3 } $ determines the values
$ \{ \phi _\mbf{p} (t _n ) \} _{t_n \in \La ; \mbf{p}
\in \bm{R} ^3} $ through the spatial Fourier transform. Here $\phi _\mbf{p} (t) := \phi (t, \mbf{p})$ is
the spatial Fourier transform of $\varphi (t, \mbf{x})$. These values together with the reconstruction formula
(\ref{eq:recon}) can be used to calculate $\phi _\mbf{p} (t)$
for all time $t$ and for all fixed spatial modes $\mbf{p}$. Taking
the inverse Fourier transform then yields the original bandlimited
field $\varphi (t, \mbf{x} )$ for all values of space and time.

\subsection{Aside: Fixed temporal modes of bandlimited fields on flat spacetime}

    The conclusions in this section followed from the assumption
that a spatial mode $\mbf{p}$ is fixed. One could instead choose to
fix a temporal mode, or any combination of the temporal mode and
spatial co-ordinate modes, and perform a similar analysis.

Note that while the situation is symmetric in $1+1$ dimensional spacetime, \emph{i.e.},
whether one fixes a temporal or a spatial mode one obtains identical results, the situation
is slightly different in $1+k$ dimensional spacetime if $k\geq 2$.

For example, in $1+2$ dimensional spacetime, the spatial bandwidth volume of
a fixed temporal mode $p_0$ is the area between two disks of radii $r_1$ and $r_2$ (see equation (\ref{eq:radii})). Moreover, this area is constant, independent of $p_0$.
If $k =3$ the spatial bandwidth volume is the volume between two spheres of radii $r_i$, and this volume
diverges as $p_0 \rightarrow \infty$. While each fixed temporal mode has a finite spatial bandwidth volume,
there is no uniform upper bound for this volume in $1+3$ dimensions.  Landau's theorem then implies that there is no set of spatial points $\La$ which
is a set of sampling for every fixed temporal mode $p_0$ in $1+3$ dimensional flat spacetime, and hence there is no
reconstruction formula of the type (\ref{eq:recon}) for a fixed temporal mode $\phi _{p_0} (\mbf{x} )$ of a bandlimited
function on this spacetime. (See \cite[Chapter 6]{Robthesis} for details.)

\section{Expanding FRW spacetimes}

We are now prepared to begin our investigation of the covariant UV cutoff in expanding spacetimes.  The line element for FRW spacetime can be expressed as
$ ds^2 = -dt^2 + a^2(t) d\mbf{x} ^2 $ where $t$ is proper time.
The d'Alembertian in $1+3$ dimensions is then a second order
differential linear operator given by \be \wt{\Box} :=
-a^{-3} (t) \left( \frac{d}{dt} a^3 (t) \frac{d}{dt} - a(t) \Delta
\right) , \ee where $\Delta$ is the spatial Laplacian.  This expression
is defined to act on a suitable dense domain in $L^2 (\bm{R} ^4 , a ^3 (t) \, dt \, d\mbf{x})$.

Here it is assumed that $a(t)>0$ is a positive
function. Under Fourier transform in the spatial variables we obtain the
differential operator
\be \wt{\Box} _k := -a^{-3} (t) \left( \frac{d}{dt} a^3 (t) \frac{d}{dt} + a(t) k^2
\right).  \label{eq:dalbert} \ee  Here $k^2 := | \mbf{k}
| ^2 \geq 0$ is the magnitude of the comoving spatial frequency vector.
In order to construct a well-defined linear operator using this expression,
we need to specify its domain.

    For example, consider de Sitter spacetime, for which $a(t)
=e^{Ht}$, where $H$ is the Hubble constant. In this case one can
show that there is no unique self-adjoint operator associated with
the expression (\ref{eq:dalbert}). Instead one can use this
expression to define a symmetric operator with equal and infinite
deficiency indices on a dense domain in the Hilbert space.  Once one has
this symmetric operator one can construct a family of self-adjoint extensions (self-adjoint
operators which extend the original symmetric operator). Each
self-adjoint extension is a possible choice of self-adjoint d'Alembertian for this
spacetime.  Physical input is necessary to determine which choice is correct.
Indeed, the covariant
ultra-violet cutoff $\Om$ is imposed by projecting onto the
subspace spanned by eigenfunctions to the d'Alembertian whose
eigenvalues are less than or equal to $\Om ^2 $ in magnitude. This means that, in principle, one must implement the correct
 self-adjoint extension of the symmetric d'Alembertian operator before cutting off
its spectrum. Fortunately, in some cases it will be sufficient to show that the choice of self-adjoint extension does not matter to answer the question at hand.

    For the remainder of this paper, the time co-ordinate will be
restricted to an interval $[t_i , t_f]$, where $-\infty \leq t_i < t_f < \infty$.
For power law spacetime, $a(t) = (Ht) ^k$, it is sensible to choose $t_i =0$ (indeed if $k$ is not an integer,
this scale factor is nonsensical for negative $t$),  but
for spacetimes such as de Sitter, where $a(t) = e^{Ht}$, or something more general such as
$a(t) = \sin ^2 (t) e^{t}$, we are free to choose $t_i = - \infty$.
The restriction to times $t_f$ that are finite is not necessary but it will be useful to determine a mode's total number of degrees up to some finite time.

Let us  now consider a comoving spatial mode $\phi(t,\mbf{k})$ and investigate, as in the
case of flat spacetime, the implications of the cutoff of the spectrum of the d'Alembertian for that mode.
In other words, we consider operators defined using the expressions given by equation (\ref{eq:dalbert}) acting
on the Hilbert space $\mc{H} := L^2 ([t_i , t_f ] ; a^3 (t) \, dt )$
for fixed $k := | \mbf{k} |$. Cutting off the spectrum of the
d'Alembertian by $\Om ^2$ on the full spacetime then amounts to
cutting off the spectrum of each of the fixed-$k$ d'Alembertians
$\dal _k$ by $\Om ^2$. As discussed above, in order to make this operation
precise, we must use the expression (\ref{eq:dalbert}) to
define a particular self-adjoint operator for each comoving spatial momentum magnitude $k$. Since (\ref{eq:dalbert})
is a second order Sturm-Liouville differential expression, we can do this by using the theory
of ordinary differential operators \cite{Naimark}:

Define the linear manifold \be
\dom{\dal _k ^* } := \{ \phi \in \mc{H} \  | \ \phi , a^3(t) \phi
' \in AC [t_i, t_f] ; \dal _k \phi \in \mc{H} \}, \ee and define $
\dom{\hat{\dal} _k  } $ as the set of all $\phi \in \dom{\dal _k ^*}$
which have support contained in a compact subinterval of $(t_i ,
t_f )$. Here, $AC [t_i, t_f]$ denotes the set of all absolutely
continuous functions on $[t_i , t_f]$.
The linear manifold $\dom{\dal _k ^*}$ is the largest
linear manifold in $\mc{H}$ on which the formal expression
(\ref{eq:dalbert}) can be defined for fixed $k$. It is not
difficult to verify that the operator $\hat{\dal} _k$ defined by
$\hat{\dal} _k \phi = \wt{\dal}  _k \phi $ for all
$\phi \in \dom{\hat{\dal} _k}$ is a symmetric operator. Now let $\Box _k$ again denote the closure of $\hat{\Box} _k$.
It follows from \cite[Section 17]{Naimark} that the domain of $\Box _k$ is given by:
\be \dom{\Box _k} =  \{ \phi \in \dom{\Box _k ^*} \ | a^3 (t_i ) \phi ' (t_i) = 0 = a^3 (t_f) \phi ' (t_f ) ; \ \phi (t_i ) = 0 =\phi (t_f) \}. \label{eq:domaink}\ee
It can be shown that the operator $\dal _k ^*$, defined by $\dal _k ^* \phi = \wt{\dal} _k
 \phi $ for all $\phi \in \dom{\dal _k ^*}$ is the adjoint to
$\dal _k$ (\cite{Naimark}, Section 17).  It also follows from \cite[Section 17]{Naimark} that $\Box _k$ has deficiency indices $(n,n)$ where
$0 \leq n \leq 2$. Recall that the deficiency indices $(n_+ , n_-)$ of a symmetric operator $A$ are defined to be $n_\pm := \dim \left(\ker{ A^* \mp i } \right)$,
that $A$ has a $\mc{U} (n)$-parameter family of self-adjoint extensions if and only if $n_+ = n = n_-$, and that if $n_+ = 0 = n_-$, then $A$
is essentially self-adjoint. Here $\mc{U} (n)$ denotes the unitary group of $n\times n$ unitary matrices. For more details of the theory of symmetric operators and their self-adjoint extensions, we refer the reader
to \cite{Glazman,Reed2}, and to \cite{Naimark} for an introduction to symmetric differential operators.

\subsection{Deficiency indices and self-adjoint extensions}

    To investigate the effects of the covariant ultraviolet
cutoff on the spacetimes discussed above, one must
cut off the spectrum of self-adjoint extensions of the symmetric
d'Alembertians $\Box _k$. The deficiency indices and therefore the set of possible self-adjoint extensions of the d'Alembertian depend on the choice of scale factor function $a(t)$ (although a few conclusions can even be drawn before making a specific choice of $a(t)$).

As we will see in the next section, the case where $\Box _k$ has deficiency indices $(2,2)$ will be generic and therefore of particular interest. If $\Box _k$ has deficiency indices $(2,2)$,
then by a theorem of Krein, the spectrum of any self-adjoint extension $\dal
_k '$ is bounded below \cite[pg. 93]{Naimark}. Furthermore, the
resolvent $(\dal _k ' - \la ) ^{-1} $, with $\la \in \mathbb{C} \setminus
\mathbb{R}$, will be a compact Hilbert-Schmidt
operator \cite[Section 19]{Naimark}. Hence, the spectrum of such
a resolvent consists of the closure of the set of eigenvalues
whose only possible accumulation point is $0$. It follows that the
spectrum of any self-adjoint extension $\dal _k ' $ of $\dal _k $
is purely discrete and has no finite accumulation point. Furthermore,
one can show that each of the symmetric operators $\Box _k $ is simple,
\emph{i.e.}, each has no eigenvalues. To see this,  first recall the formula (\ref{eq:domaink})
for the domain of $\Box _k$. Then by the existence-uniqueness theorem for ordinary differential equations, \cite[Section 16]{Naimark}, it follows that if $\phi \in \dom{\Box _k}$ is a solution to $\Box _k \phi =
\la \phi $, then $\phi \equiv 0$ since $a^3 (t_i ) \phi ' (t_i ) = 0 = \phi (t_i)$.

If $\Box _k$ is a simple symmetric operator with deficiency indices
$(2,2)$, it follows that the spectrum of any self-adjoint extension $\Box _k '$
of $\Box _k$ consists of eigenvalues of multiplicity at most $2$, and that
given any $\la \in \bm{R}$ there is a self-adjoint extension $\Box _k '$ of
$\Box _k $ for which $\la $ is an eigenvalue of multiplicity $2$ \cite[Section 83]{Glazman}.

  Imposing the covariant ultra-violet cutoff on the kind of spacetime discussed above means
that one chooses a particular self-adjoint extension $\dal _k ' $
and then considers the subspace, $B_k (\Om )$, spanned by eigenfunctions of this
d'Alembertian whose eigenvalues are less than or equal to $\Om ^2
$ in magnitude. The nature of the spectra of the $\dal
_k '$ in this kind of spacetime implies that each of the subspaces $B_k (\Om)$ will be of
finite dimension $N_k < \infty$. Let $\{ f _n \} _{n=1} ^N$ be a basis for such a
subspace. Since the $f _n$ are linearly independent functions of
time, almost all sets of  $N$ time points $\{ t_m \} _{m =1} ^N$
in $[t_i ,t_f]$ are such that the matrix $A$ with entries $A_{nm}
:= f _n (t_m ) $ is invertible. Letting $\phi _\mbf{k} = \sum _{n=1} ^N
c_n f _n$ denote the comoving spatial mode $\mbf{k}$ of a covariantly bandlimited
field in this spacetime, it follows that $ \phi _\mbf{k} (t_m) = \sum
_{n=1} ^N c_n f _n (t_m) $ so that \be \phi _\mbf{k} (t) = \sum _{n=1}
^N \sum _{j=1} ^N A ^{-1} _{nm} \phi _k (t_m)  f_n (t) .\ee Hence,
each spatial mode of a covariantly bandlimited field in any FRW
spacetime for which $\Box _k$ has deficiency indices $(2,2)$ is
completely determined by the values it takes on a
finite number of points in time, and has only a finite number of
temporal degrees of freedom. Of course this simple linear algebra argument applies to any
finite dimensional subspace of functions.

\subsubsection{The deficiency indices are generically $(2,2)$}

As we will now show, certain generically satisfied conditions are sufficient to guarantee that the deficiency indices of $\Box _k$ are $(2,2)$.

If $t_i$ is finite, then $\Box _k$ will have deficiency indices $(2,2)$  provided the function $a (t)$ is sufficiently well-behaved on the compact interval $[t_i , t_f]$ \cite[Section 17]{Naimark}.  In particular, this is the case if $a(t)$ and $1 / a(t)$ are
finite, positive and differentiable on $[t_i , t_f]$. If $t_i = -\infty$, then applying \cite[Corollary 8]{Everitt-limitcircle} shows that a sufficient condition
for $\Box _k$  to have deficiency indices $(2,2)$ is that \be \frac{1}{k} \int _{-\infty} ^{t_f} a(t) \, dt < \infty . \ee Note that we need to assume that $k \neq 0$, although this is not a significant restriction. Indeed, for fields with a finite positive mass, $m>0$, the term
$k^2$ becomes $k^2 +m^2$, eliminating the zero mode as a special case. If each $\dal _k$ has equal deficiency indices $(2,2)$, then each $\dal _k$
has a $\mc{U}(2)$-parameter family of self-adjoint extensions \cite{Glazman}. The
extensions can be constructed by extending the domain of $\dal _k$
by choosing appropriate boundary conditions at the end points of
the interval $[t_i , t_f]$ \cite[Section 18]{Naimark}.

\it Fixed assumptions about the cosmic expansion. \rm  %\label{subsubsection:fixed}
From now on we will always assume that $a(t)$ is positive, differentiable and finite. Also, unless stated otherwise, we will assume for the remainder of the paper that either:
\bn
    \item $t_i > -\infty $ so that $[t_i ,t_f]$ is a compact interval and $\Box _k$ has deficiency indices $(2,2)$ for any $k\geq 0$, or

    \item $\int _{-\infty} ^{t_f} a(t) dt < \infty$ so that $\Box _k$ has deficiency indices $(2,2)$ for any $k > 0$.

\en

\subsection{ Example : de Sitter spacetime}

    In de Sitter spacetime, $a(t) = e^{Ht}$, one can verify that in the case where $t_i = - \infty$, the operator
$\dal _k$, for $k>0$, has deficiency indices $(2,2)$:

Switching to conformal time co-ordinates, let $\eta (t) = \frac{1}{H}e^{-Ht}$, \emph{i.e.}, $\eta ' (t) = -1/a(t)$ and $\eta \in [\eta _f , \infty )$.
It follows that $a (\eta ) = \frac{1}{H \eta}$ and that the operator $\Box _k$ takes the form:

\be \Box _k = - a^{-4} (\eta) \left( \pat _\eta a^2 (\eta ) \pat _\eta + a^2 (\eta ) k^2 \right) ,\ee

\noindent acting on a suitable dense domain in $L^2 \left( [\eta _f , \infty) , a^4 (\eta ) d\eta \right)$. If they exist, any eigenfunctions of the adjoint operator $\dal_k^*$ must satisfy the following differential equation:

\be \label{eq:besselDE} \phi '' (\eta ) -\frac{2}{\eta} \phi ' (\eta ) + \left( k^2 + \frac{\la}{H^2 \eta ^2} \right) \phi (\eta ) = 0 .\ee
Two linearly independent solutions are given by the Bessel functions
\be f_\la (\eta ) = \eta ^{3/2} J _{\beta(\la)} (k \eta) , \ee and
\be g_\la (\eta ) = \eta ^{3/2} Y _{\beta(\la)} (k \eta ), \ee where $\beta(\la) := \sqrt{\frac{9}{4} - \frac{\la}{H^2}}$.
If $\la$ is such that $\beta(\la) \neq \bm{Z}$, then the second linearly independent solution can be chosen to be
$h _\la (\eta) = \eta ^{3/2} J _{-\beta (\la)} (k \eta )$ instead.

Since $\Box _k$ commutes with complex conjugation, it must have equal deficiency indices \cite[Theorem X.3]{Reed2}, so we just need to check that if $\la := \left( \frac{9}{4} -i \right) H^2 $, $\beta(\la) = e^{i\frac{\pi}{4}} = \frac{1+i}{\sqrt{2}}$, then both solutions $f_\la$ and $h_\la $ are normalizable.  This will show that $2 = n_+ = n_-$. Now for large $\eta$,
\be J_\beta (k \eta) \sim  \sqrt{\frac{ 2}{\pi k \eta }} \cos \left( k\eta - \frac{\beta\pi}{2} - \frac{\pi}{4} \right) + \mc{O} \left( \frac{1}{k\eta} \right) \ee so that
$f_\la (\eta)$ is asymptotically proportional to $ \eta \cos \left( k \eta -  c \right) $  for large $\eta$ \cite[Chapter 11.6]{Arfken}.  It follows that $f_\la$ and similarly $g_\la$ are square integrable with respect to the measure $a^4 (\eta ) d\eta = \frac{1}{H^4 \eta ^4} d\eta$  so that $\Box _k$ has deficiency indices $(2,2)$ for $k >0$, as expected.

As discussed above, it follows that the spaces $B_k (\Om )$ of fixed co-moving spatial momentum modes of bandlimited fields have finite dimension $N_k$ for any $k >0$.
It is remarkable that, therefore, any non-zero fixed co-moving spatial mode in de
Sitter spacetime with a finite end-time has only a finite number
of degrees of freedom in time $t \in (-\infty , 0 ]$. Intuitively, this is plausible because, since the spacetime is expanding
at an exponential rate, any
fixed co-moving spatial mode with magnitude $k$ corresponds to
extremely small proper wavelengths or high proper spatial
frequencies  for most of the time in $t \in (-\infty , 0]$. As was shown for the case of flat spacetime, the
larger the proper spatial frequency mode, the smaller is its
density of temporal degrees of freedom. Hence, since any fixed
co-moving spatial mode in de Sitter spacetime with a finite end-time
corresponds to exponentially large proper spatial frequencies for
most of $t \in (-\infty , 0]$, it is to be expected that such a
comoving mode could have merely a finite number of temporal
degrees of freedom.

\subsubsection{The zero mode}

As remarked previously, the zero mode is unphysical and can be safely ignored for any massive scalar field.  Nevertheless we
include some facts here about the zero mode for the sake of completeness.
In the case where $k=0$, the ODE becomes:
\be \phi '' (\eta ) -\frac{2}{\eta} \phi ' (\eta ) +  \frac{\la}{H^2 \eta ^2}  \phi (\eta ) = 0 ,\ee which has
two linearly independent solutions $f_\la (\eta ) = \eta ^{q_1 (\la)} $ and $g_\la (\eta ) = \eta ^{q_2 (\la)}$ where
$q_1 (\la) = \frac{3}{2} - \sqrt{\frac{9}{4} - \frac{\la}{H^2}}$ and $q_2 (\la) = \frac{3}{2} + \sqrt{\frac{9}{4} - \frac{\la}{H^2}}$.
Again, choosing $\la = \left( \frac{9}{4} -i \right) H^2 $ we obtain that
\be f_\la (\eta ) = \eta ^{\frac{3 +\sqrt{2}}{2} +\frac{i}{\sqrt{2}}},\ee
and \be g_\la (\eta ) = \eta ^{\frac{3 -\sqrt{2}}{2} +\frac{i}{\sqrt{2}}} .\ee
It is now clear that $f_\la $ is square integrable while $g_\la $ is not (with respect to the measure).  This shows that for de Sitter
spacetime with $t_i = -\infty$, $\Box _0$ has deficiency indices $(1,1)$. Also note that if $\la \geq \frac{9H^2}{4}$, then
the real part of both $q_i (\la)$ is equal to $3$, which shows that both solutions are non-normalizable. It then follows from \cite[Theorem 3, pg 92]{Naimark},
that $[ \frac{9H^2}{4} , \infty ) $ belongs to the continuous spectrum of $\Box _0$. Moreover the fact that $[\frac{9}{4} , \infty )$ is in the continuous spectrum
of the zero mode can be used to show that if $\Omega^2 > \frac{9}{4}$ that $N_k \rightarrow \infty$ as $k \rightarrow 0$. See for example \cite[Section 8.4.1]{Robthesis}.

Even though $B_0 (\Om )$ has an infinite number of temporal degrees of freedom, the zero mode still has a finite temporal density of degrees of freedom.
For simplicity consider $1+1$ dimensional de Sitter spacetime. Then, for an appropriate choice of self-adjoint extension of $\Box _0$, it can be shown that if $\phi \in B_0 (\Om )$ where $\Omega^2 = B^2 + \frac{1}{4}$ then
\be \phi (\eta ) = \sum _{n =0} ^\infty \phi (\eta _n ) K (\eta _n , \eta ) ,\ee where $\eta _n = e^{\frac{n}{B}}$ and
\be K(\eta _n , \eta ) = \sqrt{\frac{\eta}{\eta _n}} \left( \frac{ \sin \left( B \pi \left( \ln (\eta _n \eta ) \right) \right) }{B \pi \ln (\eta _n \eta ) }
+ \frac{\sin \left( B \pi \ln (\eta / \eta _n ) \right) }{B \pi \ln (\eta / \eta _n ) } \right)\ee \cite[Section 8.4.4, pg 99]{Robthesis}.  This shows that the zero spatial
mode of a covariantly bandlimited field is stably reconstructible from the values it takes on the set $\La := \{ \eta _n \} _{n=0} ^\infty$.

\section{The number of temporal degrees of freedom of a fixed comoving spatial mode}

We saw that under the assumptions above, $\Box _k$ has
deficiency indices $(2,2)$ for any $k>0$.  Recall that as a consequence, no matter which self-adjoint extension $\Box _k ' $ of $\Box _k$  is used to
define $B_k (\Om )$, $N_k := \dim \left( B_k (\Om ) \right) $ will be finite.

In flat spacetime we observed that large proper momentum spatial modes have less temporal bandwidth than smaller proper momentum modes.  This suggests
that also in expanding FRW spacetimes, large proper momentum spatial modes  will have fewer temporal
degrees of freedom than smaller proper momentum spatial modes. Studying this problem amounts to studying the spectrum of the second order differential operators $\dal _k$.

\subsection{The freezing of comoving modes at early times}

At early enough times, the proper wavelength of any comoving mode is arbitrarily small. From the study of the UV cutoff on Minkowski space, we therefore intuitively expect that the density of temporal degrees of freedom of a comoving mode drops for earlier and earlier times, until it reaches zero or a very small number and the mode therefore freezes.

Equivalently to considering earlier and earlier times, we may of course also consider larger and larger $k$. In this subsection, we will therefore study the behaviour of $N_k$ for large k. We will show that (under our fixed assumptions about the cosmic expansion $a(t)$) there indeed exists a $K >0$ such that $N_k = c  \leq 2 $ for all
$k \geq K$. Here $K$ is independent of the choice $\Box _k ' $ of self-adjoint extension used to define $B_k (\Om )$, while the fixed constant $c$ depends on the choice of self-adjoint extension. Again, the choice of self-adjoint extension is to be determined by physics.

Choose $\la '  \in \bm{R}$ such that $|\la ' | \leq \Omega^2$. By \cite[Section 83]{Glazman}, there is a unique self-adjoint extension $\Box _k ' $ of $\Box _k $ which
has $\la '$ as an eigenvalue of multiplicity $2$. Let $f_{\la '} $ and $g_{\la '} $ denote two linearly independent solutions to the ordinary
differential equation $\wt {\Box } _k   \phi = \la ' \phi$. Here the tilde over the $\Box _k$ is used to denote that we are not considering $\wt {\Box } _k$ as a differential
operator acting on any fixed domain, but as a differential expression acting on any functions for which this expression is defined.

\begin{lemming}
If a real number $\la \neq \la ' $ is an eigenvalue of $\Box _k ' $ then
$0 = \Delta ( \la ; \la ' , k ) := \ip{f_\la}{f_{\la '} } \ip{g_\la}{g_{\la '}} - \ip{ f_\la }{g_{\la ' }} \ip{g_\la }{f_{\la '}}$.
\end{lemming}

\begin{proof}
    If $\la $ is another eigenvalue of $\Box _k '$ then there are $c_1 , c_2 \in \bm{C}$ such that $ c_1 f_\la + c_2 g_\la $ is an eigenvector of $\Box _k '$, and hence it must be orthogonal
to both $f_{\la '}$ and $g_{\la '}$. Hence
\be \left( \begin{array}{c}  0 \\ 0 \end{array} \right) = \left( \begin{array}{c}  \ip{c_1 f_\la + c_2 g_\la}{f_{\la '}}  \\ \ip{c_1 f_\la + c_2 g_\la}{f_{\la '}} \end{array} \right)
= \left( c_1, c_2 \right) \left( \begin{array}{cc}  \ip{f_\la}{f_{\la '}} & \ip{f_\la}{g_{\la '}} \\  \ip{g_\la}{f_{\la '}} & \ip{g_\la}{g_{\la '}} \end{array} \right). \ee
It follows that the determinant of the above matrix, which is $\Delta (\la ; \la ' , k)$, vanishes.
\end{proof}

Using the method of Picard iterates that is often employed to prove the existence-uniqueness theorem for ordinary differential equations, it is not difficult to show that if
the solutions $f_\la$ and $g_\la$ are chosen by imposing fixed initial conditions at some regular point in $[t_i , t_f]$, then $f_\la (\eta)$ and $g_\la (\eta)$ are entire as functions of $\la $ for fixed $k$ and $\eta$
\cite[pages 51-56]{Naimark} and \cite[Section 2.3]{Hille}. Using this fact and Morera's theorem, it is straightforward to check that if $f_\la$ and $g_\la$ are chosen in this way,
that $\Delta (\la ) := \Delta (\la ; \la ' , k )$ for fixed $\la ' $ and $k$ is an entire function of $\la$.

Since we assume $t_f$ is finite, we can define a conformal time variable \be \eta :=  \int _t ^{t_f}  a^{-1} (t^\prime) \, dt^\prime. \ee It follows that
$\eta \in [0 , \eta _i ]$ where $\eta _i := \eta (t_i)$.  For example, in de Sitter space with $t_i = - \infty$ we would get that $\eta _i = + \infty$.
The line element in these co-ordinates is $ds^2 = a^2 (\eta ) \left( -d\eta ^2 + d\mbf{x} ^2 \right)$.  Re-calculating the symmetric d'Alembertian for a fixed co-moving spatial
mode of magnitude $k$ in the $\eta$ coordinates yields

\be \Box _k  = - a^{-4} (\eta) \left( \pat _\eta a^2 \pat _\eta + a^2 k^2 \right) ,\ee acting on a suitable dense domain in the Hilbert space
$L^2 \left( [0 , \eta _i ] ; a^4 (\eta ) \, d\eta \right)$. %Our fixed assumptions (\ref{subsubsection:fixed}) imply that $\eta _f$ is finite while $\eta _i$ is possibly equal to $+\infty$.
Then, $\wt{\Box} _k \phi (\eta ) = \la \phi (\eta )$ takes the form

\be (p \phi ' ) ' + \left( r k^2 -q \right) \phi =0 \label{eq:ode},\ee with $p(\eta) := a^2 (\eta)$, $r (\eta ) = a^2 (\eta )$, and $q(\eta ) = - a^4 (\eta ) \la$. Here we have chosen
$q(\eta )$ so as to treat $k^2$ as an `eigenvalue'. This ODE is already almost in Liouville normal form \cite[Section 2]{Fulton}.  To convert it to Liouville normal form we let $U(\eta ) := \phi (\eta )
\left( r(\eta ) p(\eta ) \right) ^{1/4} = a (\eta ) \phi (\eta )$. This yields the new ODE
\be -U'' + Q U = k^2 U \label{eq:lnf},\ee with the new potential \be Q(\eta ) = \frac{a''(\eta )}{a(\eta)} - \la a^2 (\eta ). \label{eq:potent} \ee
For large $k$, the potential $Q$ becomes negligible, and it is intuitively clear that solutions $U$ will behave asymptotically like sine or cosine. Therefore, the solutions $\phi$ to the
original ODE behave asymptotically like $a^{-1} (\eta )$ multiplied with some linear combination of sine and cosine functions.

More precisely, by \cite[Section 3]{Fulton}, let $U_\la $ and $V_\la$ be the solutions to the ODE (\ref{eq:lnf}) which obey the boundary conditions
\be \left( \begin{array}{c} U_\la (0) \\ U _\la ' (0) \end{array} \right) = \left( \begin{array}{c} 1 \\ 0 \end{array} \right) \ee and
\be \left( \begin{array}{c} V_\la (0) \\ V _\la ' (0) \end{array} \right) = \left( \begin{array}{c} 0 \\ k \end{array} \right) . \ee

The following asymptotic formulae are valid for large $k$:
\be U_\la (\eta) \sim \cos (k \eta) + \mc{O} \left( \frac{1}{k} \right) \ \ \mr{and} \ \ V_\la (\eta) \sim \sin(k \eta ) + \mc{O} \left(\frac{1}{k} \right) .\ee
The exact form of the asymptotic series for $U_\la$ and $V_\la$ depends on the potential $Q$ (equation (\ref{eq:potent})), which is determined by the scale factor
$a(\eta )$ and the eigenvalue $\la$ \cite[Section 3]{Fulton}.

Now let $f_\la = a^{-1} U _\la $ and $g_\la = a^{-1} V_\la$. Then $f_\la $ and $g_\la$ are two linearly independent solutions to the ODE (\ref{eq:ode})
which obey the boundary conditions:
\be \left( \begin{array}{c} f_\la (0) \\ f _\la ' (0) \end{array} \right) = \left( \begin{array}{c} 1 / a(0) \\ 0 \end{array} \right) \ee and
\be \left( \begin{array}{c} g_\la (0) \\ g _\la ' (0) \end{array} \right) = \left( \begin{array}{c} 0 \\  k / a(0) \end{array} \right) , \ee
and which have the asymptotic behaviour:
\be f_\la (\eta) \sim a^{-1} (\eta ) \cos (k \eta) + \mc{O} \left( \frac{1}{k} \right) \ \ \mr{and} \ \ g_\la (\eta) \sim a^{-1} (\eta ) \sin(k \eta ) + \mc{O} \left(\frac{1}{k} \right) .\ee

Now we define the function $\Delta (\la ) := \Delta (\la ; \la ' , k )$ as before using this choice of $f_\la$ and $g_\la$ for each $\la \in \bm{R}$. As discussed previously,
$\Delta (\la )$ will be an entire function of $\la$ for fixed $\la '$ and $k$. By the asymptotic formulae for $f_\la , g_\la$, it follows that as $k \rightarrow \infty$,
$\ip{f_\la}{f_{\la'}}$ approaches
\be \int _0 ^{\eta _i} \cos^2 (k \eta ) a^2 (\eta ) \, d\eta .\ee  Note that the fixed assumption that $\int _{t_i} ^{t_f} a(t) dt < \infty $ becomes
\be \int _0 ^{\eta _i} a^2 (\eta ) \, d\eta < \infty \ee upon transforming to conformal time. As $k \rightarrow \infty$ the frequency of oscillation of $\cos (k \eta )$ diverges,
and intuitively this integral will converge to the integral of $a^2 (\eta )$ multiplied by $\frac{1}{2}$, the average value of $\cos ^2 (\pi x)$ over one half period $0 \leq x \leq 1$:

\begin{lemming}
Suppose that $a^2 (\eta )$ is integrable on the interval $[0 , \eta _i]$. Then
\be \lim _{k \rightarrow \infty}  \int _0 ^{\eta _i } \cos ^2 (k \eta ) a^2 (\eta ) d\eta = \frac{1}{2} \int _0 ^{\eta _i} a^2 (\eta ) d\eta = 1/2 \| 1/a \| _\mc{H} ^2 .\ee
\label{lemming:average}
\end{lemming}

\noindent In the statement of the above lemma, $\mc{H} := L^2 \left( [0 , \eta _i ] ; a^{4} (\eta ) \, d\eta \right) $, and $\| \cdot \| _\mc{H}$ denotes the norm in this Hilbert space.

\begin{proof}
    This is a straightforward consequence of the Riemann-Lebesgue lemma. Indeed, as remarked before, our fixed assumptions about the cosmic expansion $a(t)$  imply that $1/a \in \mc{H}$,
so that $a^2 (\eta ) \in L^1 [0, \eta _i ]$, and in particular $a^2 (\eta ) \chi _{[0 , \eta _i ]} (\eta) \in L^1 (\bm{R} )$.  Using the trigonometric identity $2 \cos ^2 (x) = 1 + \cos (2x) $ yields
\be \lim _{k \rightarrow \infty}  \int _0 ^{\eta _i } \cos ^2 (k \eta ) a^2 (\eta ) d\eta  =  \frac{1}{2} \int _0 ^{\eta _i} a^2 (\eta ) d\eta + \frac{1}{4} \lim _{k \rightarrow \infty}
\intfty \left( e^{i2k \eta} + e^{-i2k\eta } \right) a^2 (\eta ) \chi _{[0, \eta _i ]} (\eta ) d\eta. \ee By the Riemann-Lebesgue lemma, the second limit vanishes.
\end{proof}

Similar arguments show that $\ip{g_\la}{g_\la '}$ converges to $ \frac{1}{2} \| 1/a \| ^2 $ and that $\ip{f_\la}{g_{\la '}}$ converges to zero. It follows that
for fixed $\la$ and $\la '$, $\Delta (\la ; \la ' , k )$ converges to $\frac{1}{4} \| 1/ a \| ^4$ as $k \rightarrow \infty$. This fact can be used to prove the main result of this section:

\begin{thm}
    Given any choice of self-adjoint extension $\Box _k '$ of $\Box _k$, let $B_k (\Om ) := \ran {\chi _{[-\Omega^2 ,\Omega^2 ] } (\Box _k ' ) }$, which has finite dimension $N_k$.
    There is a $K > 0$ such that $ N_k = c \leq 2$ for all $k \geq K$ and for any choice of self-adjoint extension $\Box _k '$ used to define $B_k (\Om)$.
\end{thm}

\begin{proof}

Let $\Omega$ be the bandlimit, or ultra-violet cutoff. Then for $\la   \in [-\Omega^2 , \Omega^2 ]$ and fixed $\la ' \in [-\Omega^2 ,\Omega^2]$, it follows that $\ip{f_\la }{f_{\la '} }$ converges
uniformly to $1/2 \| 1/ a \| ^2$ for all $\la$ in this interval as $k \rightarrow \infty$.

By Lemma \ref{lemming:average}, $\Delta (\la ; \la ' , k )$ converges to the positive constant $\frac{1}{4} \| 1/ a \| ^4 $ as $k \rightarrow \infty$ for any fixed $\la , \la '$.
Now keeping $\la '$ fixed, consider $\Delta _k (\la) := \Delta (\la ; \la ' , k )$ as a net of continuous (in fact entire) functions of the variable $\la$. Given any compact interval
$I$, since $\Delta _k$ converges pointwise to a continuous (constant) function as $k \rightarrow \infty$, this convergence must be uniform on $I$.

Hence there is a $K >0 $ such that for all $k \geq K$ we have that $\Delta (\la ; \la ' , k) > 0$ for all $ | \la | \leq \Omega^2$. This shows that
 $N_k =2 $ for all $k \geq K$ (provided we define $B_k (\Om )$ using the self-adjoint extension $\Box _k ' $ which has $\la '$ as an eigenvalue of multiplicity $2$).

Now repeat the same argument again with $\la ' = 3 \Omega^2 + \eps $ for some $\eps > 0$. Choose $K> 0 $ so large that for $k \geq K$ we have that $\Box _k ' $ has no eigenvalues
in $[-3\Omega^2 , 3\Omega^2 ]$.  Then by \cite[Theorem 3]{Martin-uncer},  $\Delta \left( \Box _k \right) _t > 2 \Omega^2 $ for all $t \in [-\Omega^2 , \Omega^2 ]$. Here $\Delta \left( \Box _k \right) _t $ denotes
the minimum uncertainty or standard deviation of the symmetric operator $\Box _k$, taken over all states with expectation value $t$.  It then follows from
\cite[Theorem 2]{Martin-uncer} that any self-adjoint extension of $\Box _k$ has at most $2$ eigenvalues in the interval $[-\Omega^2 , \Omega^2 ]$.  This demonstrates that no matter which self-adjoint
extension is used to define $B_k (\Om )$, there is a $K > 0$ independent of the choice of self-adjoint extension such that $k > K$ implies that $N_k = c \in \{ 0 ,1 ,2 \}$.
\end{proof}
\noindent Finally, we note that there always exists a self-adjoint extension of $\Box _k$ so that $N_k$ becomes $0$ for sufficiently large $k$.

\subsection{Example: de Sitter spacetime}

In the case of de Sitter spacetime, we can obtain more information about the value $K$
for which any co-moving spatial mode with magnitude $k\geq K$ has at most $2$ degrees of freedom in time. Consider a de Sitter spacetime with scale factor $a(t) = e^{Ht}$ that began expanding infinitely long ago in the past. We will consider the evolution up to a finite proper time $t_f$. In this case it is convenient to modify the calculations in the previous section by defining conformal time in a different way. Here, define conformal time by $\eta (t) = \frac{1}{H}e^{-Ht}$ so that
$\eta ' (t) = - 1/ a(t)$ as before, but now $\eta \in [\eta _f , \infty)$, where $\eta _f > 0$ is finite.

With these definitions, we can choose as our two linearly independent solutions to $\Box _k ^* \phi = \la \phi$
\be f_\la (\eta ) = \sqrt{\frac{k\pi}{2}} \eta ^{\frac{3}{2}} J _{\beta(\la)} (k \eta )  \sim \ \ \eta  \cos \left( k \eta - \frac{\pi}{2} \beta(\la) - \frac{\pi}{4} \right) + \mc{O} \left( \frac{1}{k\eta} \right)
\label{eq:asympf}, \ee
and
\be g_\la (\eta ) = \sqrt{\frac{k\pi}{2}} \eta ^{\frac{3}{2}} Y _{\beta(\la)} (k \eta )  \sim \ \ \eta  \sin \left( k \eta - \frac{\pi}{2} \beta(\la) - \frac{\pi}{4} \right) + \mc{O} \left( \frac{1}{k\eta} \right),
\label{eq:asympg} \ee

\noindent where $\beta(\lambda) \defeq \sqrt{\frac{9}{4} - \frac{\lambda}{H^2}}$. Using properties of Bessel functions, one can show that $\Delta (\la) = \Delta (\la ; \la ' , k ) $ is still entire as a function of $\la$ for fixed $\la '$ and $k$, and that it
converges uniformly to $\frac{1}{4} \| 1/ a \| ^4$ as $k \rightarrow \infty$ for $\la$ in any fixed compact interval $I$, where
\be \| 1 /a \| ^2 = \frac{1}{H^2} \int_{\eta _f } ^\infty \frac{1}{\eta ^2 } d\eta = \frac{1}{H^2 \eta _f } = a (\eta _f ) / H. \ee

One can estimate the threshold $K$ beyond which modes are dynamically frozen from the following considerations. For simplicity, fix $\lambda^\prime = 0$, and consider the problem of finding the threshold for the self-adjoint extension $\Box_k^\prime$ for which $\lambda^\prime = 0$ is an eigenvalue of multiplicity 2. Recall from Lemma 1 that $\Box_k^\prime$ has no other eigenvalues in $[-\Om^2,\Om^2]$ provided that $\Delta(\lambda;\lambda^\prime=0,k)$ has no zeros in $[-\Om^2,\Om^2]$. %In fact, the ODE $\Box _k ^* \phi = \la \phi$ provides an upper bound on the values that $\lambda$ can take. Namely, whenever $\lambda > \frac{9H^2}{4}$, the order of the Bessel functions in the mode functions \eqref{eq:asympf} and \eqref{eq:asympg} becomes imaginary. Consequently, they are not part of the Fourier basis since for early times they behave like growing and decaying exponentials.
Thus, we wish to determine the value $K$ such that whenever $k \geq K$, $\Delta(\lambda; 0, k) \neq 0$ for all $\lambda \in [-\Om^2, \Omega^2]$.

The inner products that define $\Delta(\lambda; 0, k)$ are given by
\begin{equation}
\begin{array}{ll}
\displaystyle \langle f_\lambda, f_0 \rangle = \frac{k}{H^4} \left( \frac{\sin \delta}{\lambda/H^2} + \frac{\pi}{2} F(J,J)  \right) & \displaystyle \langle f_\lambda, g_0 \rangle = \frac{k}{H^4} \left( - \frac{\cos \delta}{\lambda/H^2} + \frac{\pi}{2} F(J,Y)  \right) \medskip \\
\displaystyle \langle g_\lambda, f_0 \rangle = \frac{k}{H^4} \left( \frac{\cos \delta}{\lambda/H^2} + \frac{\pi}{2} F(Y,J)  \right) & \displaystyle \langle g_\lambda, g_0 \rangle = \frac{k}{H^4} \left( \frac{\sin \delta}{\lambda/H^2} + \frac{\pi}{2} F(Y,Y)  \right),
\end{array}
\end{equation}

\noindent where $J$ and $Y$ are placeholders for the Bessel $J$ and Bessel $Y$ functions respectively and where
\begin{equation}
F(A,B) \defeq k \eta_f \frac{A_{\beta(\lambda)-1}(k\eta_f) \, B_{3/2} (k\eta_f) - A_{\beta(\lambda)}(k\eta_f) \, B_{1/2}(k\eta_f)}{\lambda/H^2} + \frac{A_{\beta(\lambda)}(k\eta_f) \, B_{\beta(\lambda)}(k\eta_f)}{\beta(\lambda) + 3/2}
\end{equation}

\noindent for $A, B \in \{J,Y\}$. Notice that if we divide $\Delta(\lambda;0,k)$ by $k^2/H^8$ and define new variables $y \defeq k\eta_f$ and $\ell \defeq \lambda/H^2$, we have that
\begin{equation}
\frac{H^8}{k^2} \Delta(\lambda; 0, k) = G(y,\ell),
\end{equation}

\noindent a function of only two variables. Since $H \neq 0$, $\Delta(\lambda)$ and $\frac{H^8}{k^2}\Delta(\lambda)$ have the same zeros. So, equivalently to our initial problem, we may search for the threshold $Z$ such that $\frac{H^8}{k^2}\Delta(\lambda) = G(y,\ell)$ has no zeros on $\ell \in [-\frac{\Om^2}{H^2}, \frac{\Om^2}{H^2}]$ for all $y \geq Z$. Hence, let
\begin{equation}\label{eq:threshZ}
Z\left(\frac{\Omega}{H}\right) \defeq \min_{y>0} \left\{ y \; \left| \; G (y,\ell) \neq 0 \quad \forall \, \ell \in \left[-\frac{\Om^2}{H^2}, \frac{\Om^2}{H^2} \right]  \right. \right\}.
\end{equation}

\noindent Since $y = k\eta_f$, we therefore have that the threshold $K$ is given by
\begin{equation}\label{eq:threshK}
K = \frac{1}{\eta_f} Z\left( \frac{\Omega}{H} \right).
\end{equation}

Without knowing the functional form of $Z$, equation \eqref{eq:threshZ} indicates that $Z$ is a monotonically increasing function of $\Omega/H$. Since $\Delta$ approaches a constant function as $k \rightarrow \infty$ (and hence as $y \rightarrow \infty$), the $\ell$-zeros of $G(y,\ell)$ are pushed farther and farther away from the origin as $y$ increases. Increasing $\Omega/H$ widens the interval over which $\Delta$ must be nonzero, so larger values of $\Omega/H$ produce larger thresholds $Z(\Omega/H)$.

Numerical analysis corroborates the preceding conclusion; a plot of $Z(\Omega/H)$ with $\ell$ restricted to lie in $\ell \in [-\frac{\Om^2}{H^2}, \, \min\{\frac{9}{4}, \frac{\Om^2}{H^2}\}]$ is shown in Figure \ref{fig:Zplot}. When $\ell \geq \frac{9}{4}$, the orders of the Bessel functions in the definition of $Z$ become imaginary, which considerably complicates their numerical analysis. We have circumstantially observed, however, that any $\ell$-zeros of $G(y,\ell)$ for which $\ell > \frac{9}{4}$ are greater in magnitude than the first zero for which $\ell \leq \frac{9}{4}$. %Furthermore, it is unclear whether the solutions $f_\lambda$ and $g_\lambda$ are physically viable when their order is imaginary, so it may very well be that our numerical analysis is entirely accurate. \marginpar{Hi Aidan, Achim - I think this sentence beginning `Furthermore, it is unclear...' sounds a bit speculative. I would suggest we remove it.}
Although it would be desirable to verify the behaviour of $Z$ for larger values of $\Omega/H$, the computational task becomes highly nontrivial as $\Omega/H$ increases. The current numerical analysis suggests that $Z$ is proportional to $(\Omega/H)^2$.

\begin{figure}[h]
\begin{center}
\includegraphics[scale=0.62]{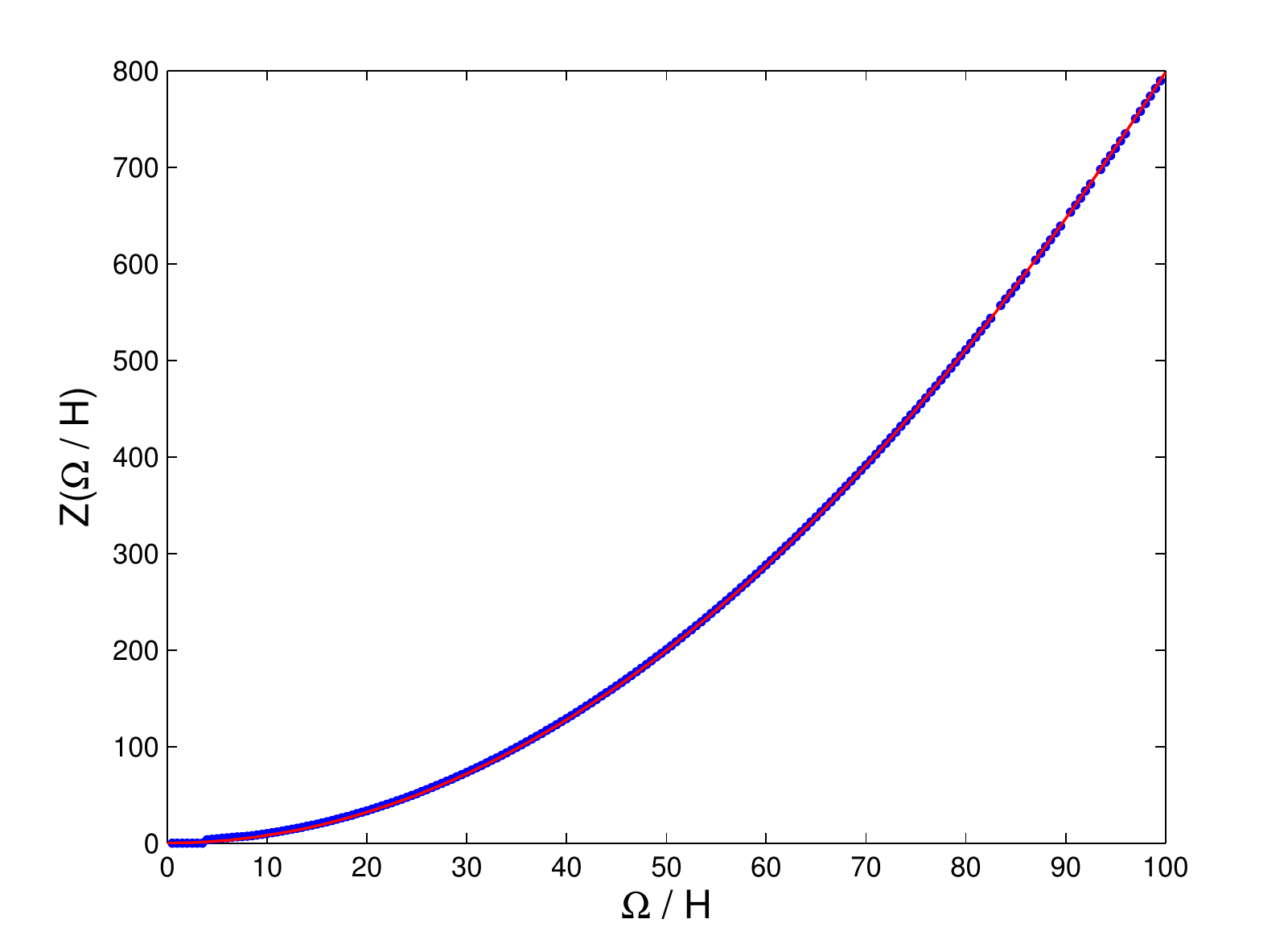}
\caption{Numerical simulation of $Z(\Omega/H)$. The red curve is a fit to the data by the curve $Z = 0.07988 (\Omega/H)^2$.}
\label{fig:Zplot}
\end{center}
\end{figure}

One would expect that the modes in de Sitter space that are dynamical are those modes that cross the Planck length before the end of inflation. Conversely, the dynamically frozen modes should be those whose comoving wavelengths are so small that they never cross the Planck length during the period of inflation. The dependence of the threshold for mode freezing $K$ on the parameters $\eta_f$, $\Omega$, and $H$ in equation \eqref{eq:threshK} is indeed consistent with this expectation.

As $\eta_f$ decreases toward 0, $K$ increases; this is natural, as a small conformal end time $\eta_f$ corresponds to a late proper end time $t_f$. The later the proper end time, the longer the period of inflation lasts, so the smaller a mode's comoving wavelength (and thus the larger $k$) must be if it is to never cross the Planck length and become dynamical.

The dependence of $K$ on the dimensionless ratio of the two length scales, $\Omega/H$, has an intuitive interpretation. Holding $H$ fixed, increasing $\Omega$ essentially corresponds to decreasing the value of the Planck length. This permits modes that were previously too small to cross the Planck length during inflation to make that crossing. Hence, the comoving wavelengths of modes that do not cross the Planck length decreases, or in other words, the threshold increases.

%The asymptotic approximations $f _\la (\eta ) \approx \cos \left( k \eta - \frac{\pi}{2} p (\la) - \frac{\pi}{4} \right) $ and
%$ g _\la (\eta ) \approx \sin \left( k \eta - \frac{\pi}{2} p (\la) - \frac{\pi}{4} \right)$ are valid for $8k \eta \gg 4 p (\la) ^2 -1  $ \cite[Chapter 11.6]{Arfken}. It follows that $N_k$ will achieve its minimum
%value for $k \gg \max _{|\la | \leq \Omega^2 } \frac{1}{\eta _f } \left| 1 - \frac{\la}{2 H^2 } \right| $.

\section{The impact of the covariant ultraviolet cutoff on the spectrum of quantum field fluctuations}

So far, we discussed how the presence of the covariant ultra-violet cutoff in expanding FRW spacetimes would affect the kinematics of individual modes of scalar fields on the spacetime. Let us now begin to investigate how the covariant cutoff would affect the full dynamics of scalar quantum fields. In particular, the presence of this natural cutoff in nature could manifest itself with a potentially measurable effect in the cosmic microwave background (CMB). It is, therefore, of great interest to implement  the covariant ultraviolet cutoff in the standard model of cosmic inflation.

In order to determine how the covariant ultra-violet cutoff affects the dynamics of a scalar quantum field, it is convenient to use the path integral formalism since it is manifestly covariant. In the path integral picture, the assumption that there exists this covariant ultraviolet cutoff in nature means that the set of scalar fields that one integrates over  in the quantum field theoretic path integral is restricted to the space $B(M, \Omega)$ of covariantly bandlimited scalar fields on the spacetime $M$.

\subsection{The two-point function}
The dynamics of a free scalar field $\hat \phi$  are determined by the Feynman propagator, which is given in the path integral picture through
\begin{equation}
G_F(x,x^\prime) = \frac{\int \phi(x) \phi(x^\prime) e^{iS[\phi]} \, \mathcal{D}[\phi]}{\int e^{iS[\phi]} \, \mathcal{D}[\phi]},
\end{equation}
\noindent and in the interaction picture of the operator formalism through:
\begin{equation}
G_F(x,x^\prime) = \bra{0} T \hat \phi(x) \hat \phi(x^\prime) \ket{0}
\end{equation}
Here, $G_F(x,x^\prime)$ is ambiguous up to a choice of vacuum which must be made on the basis of physical input.

Denoting the coordinates $x = (t, \mathbf{x})$, the key object of interest for inflationary predictions for the CMB is $G_F(t=t^\prime,\mathbf{p})$, the equal-time spatial Fourier transform of the two-point function. This is because $G_F(t=t^\prime,\mathbf{p})$ yields the fluctuation spectrum of the modes of a quantum field $\hat \phi$, and it is this type of fluctuation spectrum that determines the spectrum of the scalar and tensor fluctuations that are imprinted in the temperature and polarization spectra of  the CMB \cite{Mukhanov}.
%Returning to the path integral picture, a quick way to arrive at the standard equation of motion for $G_F(x,x^\prime)$ is as follows. Consider carrying out a path integral over the following total derivative:
%
%\begin{equation}
%0 = \int \frac{\delta}{\delta \phi(x)} \left( \phi(x^\prime) e^{iS[\phi]} \right) \mathcal{D}[\phi]
%\end{equation}
%
%\noindent Using the chain rule to expand the integrand gives
%
%\begin{equation}\label{eq:eom_deriv1}
%0 =  \int \left( \delta^4(x-x^\prime) e^{iS[\phi]} + \phi(x^\prime) i \frac{\delta S[\phi]}{\delta \phi(x)} e^{iS[\phi]} \right) \, \mathcal{D}[\phi]
%\end{equation}
%
%\noindent For the Klein-Gordon action
%
%\begin{equation}
%S_{KG}[\phi] = -\frac{1}{2} \int \left( g^{\mu\nu} \phi_{,\mu} \phi_{,\nu} + m^2 \phi^2 \right) \sqrt{-g} \, d^4 x
%\end{equation}
%
%\noindent equation \eqref{eq:eom_deriv1} reads
%
%\begin{equation}\label{eq:eom_deriv2}
%0 = \int \left( \delta^4(x-x^\prime) + i \phi(x^\prime) (\Box_x - m^2) \phi(x) \right) e^{iS[\phi]} \mathcal{D}[\phi]
%\end{equation}
%
%\noindent where $\Box_x$ is the d'Alembertian with respect to unprimed co-ordinates. Separating terms in equation \eqref{eq:eom_deriv2} and dividing through by $\int e^{iS[\phi]} \, \mathcal{D}[\phi]$, we arrive at the equation of motion for $G_F$,
In the equation of motion for $G_F$,
\begin{equation}\label{eq:eom}
(\Box_x - m^2) G_F(x,x^\prime) = i \delta^4 (x-x^\prime),
\end{equation}
\noindent $\Box_x$ is the d'Alembertian with respect to unprimed co-ordinates.  The presence of homogeneous solutions requires a  choice of boundary condition on $G_F(x,x^\prime)$, which in turn corresponds to the choice of the vacuum.
Technically, the choice of boundary condition fixes a self-adjoint extension $\Box^\prime$ of $\Box$.
%In functional analytic language, this Green's function $G_F$ is the integral kernel of the inverse of $\Box$. However $\Box$ may have kernel, and the d'Alembertian is only symmetric and not self-adjoint. Consequently, specifying $G_F$ uniquely amounts to choosing a fixed self-adjoint extension $\Box^\prime$ of $\Box$, and then if $\Box^\prime$ has kernel, choosing an inverse operator $\Box^{-1}$ such that $\Box^{-1}\Box = \mathbb{1}$. This ambiguity in the choice of Green's function is equivalent to the ambiguity in the choice of vacuum in the Heisenberg picture.
\subsection{The two-point function in Minkowski space}
We can here only begin our study of the impact of the covariant cutoff on the spectrum of quantum field fluctuations, namely by considering the simple case of $1+3$ dimensional Minkowski space, where the d'Alembertian reads $\Box_x = - \frac{\partial^2}{\partial t^2} + \frac{\partial^2}{\partial x^2} + \frac{\partial^2}{\partial y^2} + \frac{\partial^2}{\partial z^2}$. In this case, since $\Box_x = \Box_{x-x^\prime}$, and with equation \eqref{eq:eom}, we can treat $G_F$ as a function exclusively of the separation $x-x^\prime$. Fourier transforming equation \eqref{eq:eom} with respect to $x-x^\prime$, yields
\begin{equation}\label{eq:fourierG}
G_F(p) = \frac{i}{(2\pi)^2} \frac{1}{p_0^2  - |\mathbf{p}|^2 - m^2 + i\epsilon}.
\end{equation}
\noindent Through the Fourier transform, the choice of boundary condition becomes a choice of pole prescription. The introduction of  Feynman's $i \epsilon$ which implies the limit $\epsilon \rightarrow 0^+$ after integrations is the well-known choice of pole prescription that yields the Feynman propagator on Minkowski space. One recovers the equal time and momentum dependent two-point function of interest by performing the inverse Fourier transform with respect to $p_0$:
\begin{equation}\label{eq:transback}
G_F(t-t^\prime, \mathbf{p}) = \frac{i}{(2\pi)^{5/2}} \int_{-\infty}^{\infty} dp_0 \, \frac{e^{i p_0 (t-t^\prime)}}{p_0^2 - |\mathbf{p}|^2 - m^2  + i \epsilon}
\end{equation}
\noindent Setting $t = t^\prime$, using standard methods of contour integration and defining $\omega \defeq \sqrt{|\mathbf{p}|^2 + m^2}$, one has:
\begin{equation}\label{eq:reg2point}
G_F(t = t^\prime, \mathbf{p}) =  \frac{1}{(2\pi)^{3/2}} \frac{1}{2\omega},
\end{equation}
\noindent
How does the above calculation change if one assumes the covariant cutoff? Let $G_F^c(x,x^\prime)$ be the composition of $G_F(x,x^\prime)$ and the projector $P_{B(M,\Omega)}$ that projects onto $B(M,\Omega)$, the space of covariantly bandlimited functions, so that
\begin{equation}
G_F^c(x,x^\prime) \defeq P_{B(M,\Omega)} G_F(x,x^\prime) = \frac{\int_{B(M,\Omega)} \phi(x) \phi(x^\prime) e^{iS[\phi]} \, \mathcal{D}[\phi]}{\int_{B(M,\Omega)} e^{iS[\phi]} \, \mathcal{D}[\phi]}.
\end{equation}
\noindent If we perform a full inverse Fourier transform on equation \eqref{eq:fourierG}, we obtain the following integral representation of $G_F(x-x^\prime)$:
\begin{equation}\label{eq:fullTransback}
G_F(x-x^\prime) = \frac{i}{(2\pi)^4} \int dp_0 \, d^3\mathbf{p} \, \frac{e^{i ( p_0 (t-t^\prime) - \mathbf{p}\cdot(\mathbf{x}-\mathbf{x}^\prime) )}}{p_0^2 - |\mathbf{p}|^2 - m^2  + i \epsilon}
\end{equation}
\noindent Recall that the plane waves $e^{i ( p_0 (t-t^\prime) - \mathbf{p}\cdot(\mathbf{x}-\mathbf{x}^\prime) )}$ are the eigenfunctions of $\Box_{x-x^\prime}$ with corresponding eigenvalues $p_0^2 - |\mathbf{p}|^2$. Equation \eqref{eq:fullTransback} is therefore manifestly a linear combination of eigenfunctions of the d'Alembertian, and so the action of the projector $P_{B(M,\Omega)}$ gives
\begin{equation}
G^c_F(x-x^\prime) = \frac{i}{(2\pi)^4} \int_{|p_0^2 - |\mathbf{p}|^2| \leq \Omega^2} dp_0 \, d^3\mathbf{p} \, \frac{e^{i ( p_0 (t-t^\prime) - \mathbf{p}\cdot(\mathbf{x}-\mathbf{x}^\prime) )}}{p_0^2 - |\mathbf{p}|^2 - m^2  + i \epsilon}.
\end{equation}
\noindent Performing a spatial Fourier transform on the previous equation to obtain the two-point function of cosmological interest, we find that
\begin{equation}
G_F^c(t-t^\prime, \mathbf{p}^\prime) = \frac{i}{(2\pi)^{11/2}} \int_{\mathbb{R}^3} d^3(\mathbf{x} - \mathbf{x}^\prime) \, e^{i\mathbf{p}^\prime \cdot (\mathbf{x}-\mathbf{x}^\prime)} \int_{|p_0^2 - |\mathbf{p}|^2| \leq \Omega^2} dp_0 \, d^3\mathbf{p} \, \frac{e^{i ( p_0 (t-t^\prime) - \mathbf{p}\cdot(\mathbf{x}-\mathbf{x}^\prime) )}}{p_0^2 - |\mathbf{p}|^2 - m^2  + i \epsilon}
\end{equation}
\noindent whence
\begin{equation}\label{eq:cutoffG}
G_F^c(t=t^\prime, \mathbf{p}) = \frac{i}{(2\pi)^{5/2}} \int_{\mathcal{I}(\mathbf{p})} dp_0 \, \frac{1}{p_0^2 - |\mathbf{p}|^2 - m^2 + i \epsilon}.
\end{equation}
\noindent The interval $\mathcal{I}(\mathbf{p})$ has two qualitatively different forms depending on the value of $|\mathbf{p}|$. Referring to equation (3) and Figure \ref{figure:bandwidth}, we have that
\begin{enumerate}[I.]
\item if $|\mathbf{p}| \leq \Omega$, then $\mathcal{I}(\mathbf{p}) = \left[-\sqrt{|\mathbf{p}|^2 + \Omega^2}, \sqrt{|\mathbf{p}|^2 + \Omega^2}\right]$, or
\item if $|\mathbf{p}| > \Omega$, then $\mathcal{I}(\mathbf{p}) = \left[-\sqrt{|\mathbf{p}|^2 + \Omega^2}, -\sqrt{|\mathbf{p}|^2 - \Omega^2} \right] \cup \left[\sqrt{|\mathbf{p}|^2 - \Omega^2}, \sqrt{|\mathbf{p}|^2 + \Omega^2}\right]$.
\end{enumerate}
\noindent Evaluating the integral in equation \eqref{eq:cutoffG}, we ultimately find that
\begin{equation}\label{eq:GFc}
G_F^c(t=t^\prime,\mathbf{p}) = \left\{
\begin{array}{ll}
\displaystyle \frac{1}{(2\pi)^{3/2}} \frac{1}{2\omega} - \frac{i}{(2\pi)^{5/2}} \frac{1}{\omega} \ln \left| \frac{r_2+\omega}{r_2-\omega} \right| & |\mathbf{p} | \leq \Omega \medskip \\
\displaystyle \frac{1}{(2\pi)^{3/2}} \frac{1}{2\omega} - \frac{i}{(2\pi)^{5/2}} \frac{1}{\omega} \left( \ln \left| \frac{r_2+\omega}{r_2-\omega} \right| - \ln \left| \frac{\omega+r_1}{\omega-r_1} \right| \right) & |\mathbf{p} | > \Omega
\end{array} \right.
\end{equation}
\noindent where $r_2 \defeq \sqrt{|\mathbf{p}|^2 + \Omega^2}$ and $r_1 \defeq \sqrt{|\mathbf{p}|^2 - \Omega^2}$. Details of the aforementioned calculation may be found in the appendix. A plot of the absolute values of both $G_F(t=t^\prime,\mathbf{p})$ and $G_F^c(t=t^\prime,\mathbf{p})$ are shown in Figure \ref{fig:plotComparison}.
\begin{figure}[h]
\begin{center}
\includegraphics[scale=0.6]{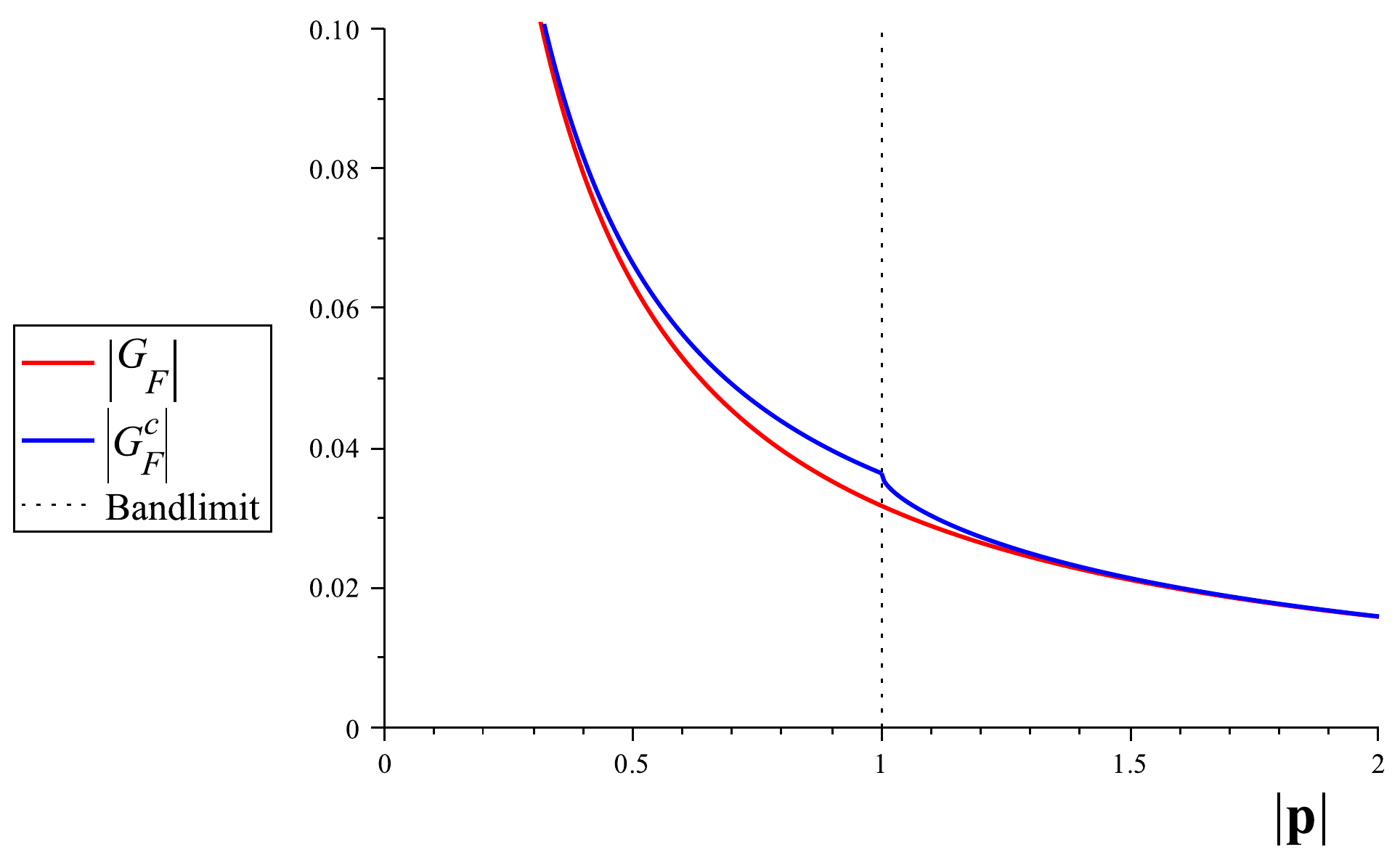}
\caption{Plots of the two-point function $G_F$ (red) and the covariantly bandlimited two-point function $G_F^c$ (blue). The plotting parameters are $\Omega = 1$ and $m = 0.01$.}
\label{fig:plotComparison}
\end{center}
\end{figure}
Notice that the integral in equation \eqref{eq:cutoffG} (which defines $G_F^c(t=t^\prime, \mathbf{p})$) is the same integral as in equation \eqref{eq:transback} (which defines $G_F(t=t^\prime,\mathbf{p})$) less an integration over $\mathbb{R} \setminus \mathcal{I}(\mathbf{p})$, \emph{i.e.},
\begin{equation}\label{eq:cutoffCorrect}
G_F^c(t=t^\prime, \mathbf{p}) = G_F(t=t^\prime, \mathbf{p}) - \frac{i}{(2\pi)^{5/2}} \int_{\mathbb{R} \setminus \mathcal{I}(\mathbf{p})} dp_0 \, \frac{1}{p_0^2 - \omega^2 + i \epsilon}.
\end{equation}
\noindent The poles of these integrands are at the points $p_0 = \pm \omega$, which always lie in the interval $\mathcal{I}(\mathbf{p})$, hence the integral over $\mathbb{R} \setminus \mathcal{I}(\mathbf{p})$ always yields a real number. Consequently, the difference between $G_F$ and $G_F^c$ is always a purely imaginary correction term, as is evident in equations \eqref{eq:GFc} and \eqref{eq:cutoffCorrect}. Since $G_F(t=t^\prime,\mathbf{p})$ is a real-valued function, the magnitude of $G_F^c(t=t^\prime,\mathbf{p})$ will always be larger for all $\mathbf{p}$, which is indeed observed in Figure \ref{fig:plotComparison}.
The physical consequence of this effect is that the covariantly bandlimited theory predicts larger quantum fluctuations of the field $\hat \phi$ than the standard theory. Intuitively, the discarding of high-frequency contributions to the two-point function eliminated destructive interference caused by these modes.

In fact, from Figure \ref{figure:bandwidth} it had to be expected that there will be a cusp in the graph of $G^c_F$ in Figure \ref{fig:plotComparison}. Referring to Figure \ref{figure:bandwidth}, the tangent line to the horizontally opening hyperbola that bounds the region $|p_0^2 - |\mathbf{p}|^2 | \leq \Omega^2$ is vertical at the cutoff value $|\mathbf{p}| = \Omega$. This point is the boundary across which the form of $\mathcal{I}(\mathbf{p})$ qualitatively changes and is the location of the cusp.
\section{Conclusions and Outlook}

We investigated the consequences of a possible covariant minimum wavelength in nature in the form of a cutoff on the spectrum of the d'Alembertian. In this scenario, wavelengths smaller than the Planck length do exist but the dynamics of such modes is in effect frozen  due to their exceedingy small temporal bandwidth. The information density in such modes is low in a literally information-theoretic sense.

In particular, we showed that comoving modes in expanding spacetimes unfreeze, i.e., that they grow temporal degrees of freedom and develop nontrivial dynamics, only after their proper wavelength starts exceeding the Planck length. Later, as in standard inflationary cosmology, once a comoving mode has outgrown the Hubble horizon it again loses its dynamics, namely because its two oscillatory solutions turn into a constant and a quickly decaying solution. (As in the usual inflationary scenario, the mode then re-acquires nontrivial dynamics after inflation ends and the mode re-enters the Hubble horizon.)

The fact that soon after their unfreezing the modes' dynamics  freezes up again, namely upon Hubble horizon crossing, could be crucial. In implies that, in principle, precision  measurements of the CMB and its polarization may provide evidence for or against the existence of a covariant bandlimit in nature. This in turn would imply the enticing prospect of experimental access to Planck scale physics. Indeed, according to the standard model of inflationary cosmology, the quantum fluctuations that seeded fluctuations in the CMB's temperature and polarization spectrum (and therefore cosmic structure formation) were frozen upon their Hubble horizon crossing during inflation. The Hubble horizon during inflation, however, is generally thought to have been only five or six orders of magnitude larger than the Planck length. This means that the physics of the Planck scale, such as potentially a covariant UV cutoff, could conceivable have an imprint in the CMB that is not all too much suppressed. This is because the magnitude of the effect can be expected to be proportional to some power, $\alpha$,  of the dimensionless ratio, $\sigma$,  of the two basic length scales, the Planck length and the Hubble length during inflation. Here, as we mentioned, $\sigma$ is known to be no larger than about $10^{-5}$ in most realistic models of inflation. In order to determine the experimental prospects, it will be important to determine $\alpha$.

\noindent To obtain a first indication of the size of the effects, let us consider Minkowski space. Here, the effect is very small, indicating a value of $\alpha$ close to 2. To see this, let us consider  Figure \ref{fig:effectStr} which shows the relative difference between $G_F(t=t^\prime,\mathbf{p})$ and $G_F^c(t=t^\prime,\mathbf{p})$. At $|\mathbf{p}| = 10^{-5} \Omega$, experiments would need to be able to measure the two-point function to the tremendous  precision of $2 \times 10^{-9} \, \%$ or better to be sensitive to the covariant cutoff.
\begin{figure}[h]
\begin{center}
\includegraphics[scale=0.5]{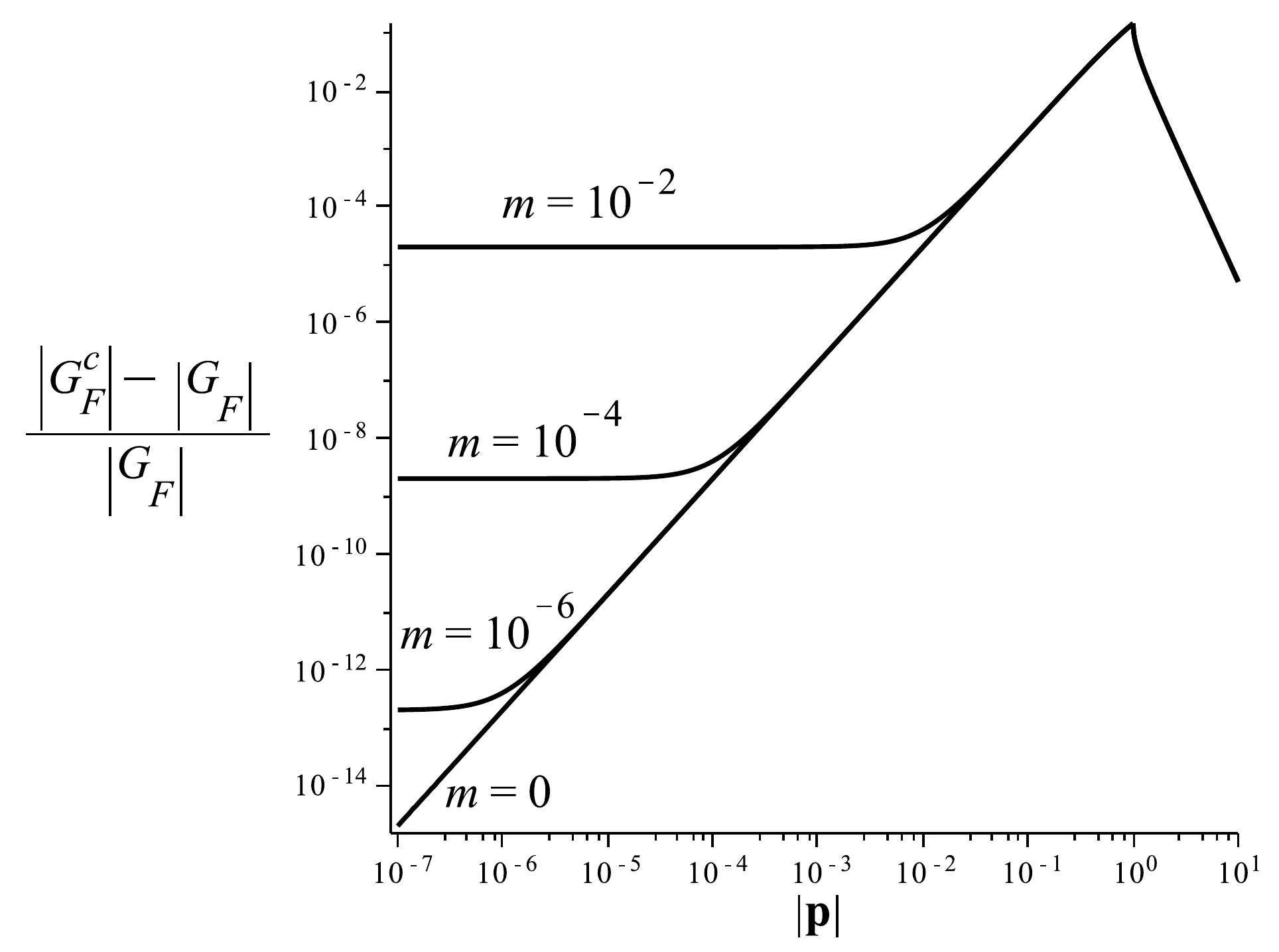}
\caption{Plot of the relative difference in magnitudes of $G_F(t=t^\prime, \mathbf{p})$ and $G_F^c(t=t^\prime, \mathbf{p})$ in Minkowski space for different field masses $m$. The cutoff is set to $\Omega = 1$.}
\label{fig:effectStr}
\end{center}
\end{figure}
For a typical inflationary spacetime such as de Sitter spacetime, or a more realistic cosmological model such as power law inflation, the effect of a covariant cutoff may well be much larger than in the case of Minkowski space, and $\alpha$ may be closer to $1$,  given that inflation tends to amplify quantum fluctuations.

It is difficult, therefore, to predict the value of $\alpha$ and therefore to determine how large an effect the covariant cutoff possesses in an expanding FRW spacetime at this stage. A detailed calculation of $\alpha$ for the covariant UV cutoff in the inflationary scenario is progress.

\bf Acknowledgment \rm

This work has been partially supported by the Discovery and Canada Research Chairs Programs of the Natural Sciences and Engineering Research Council (NSERC) of Canada and by the National Research Foundation of South Africa.

\appendix
\section{Appendix: Two-point function calculation details}

In this appendix, we discuss the details of how we calculate the covariantly bandlimited two-point function in Minkowski space given by equation \eqref{eq:cutoffG}. Recall that the region of integration $\mathcal{I}(\mathbf{p})$ assumes one of two qualitatively different forms depending on whether $|\mathbf{p}| \leq \Omega$ or $|\mathbf{p}| > \Omega$. We must consider each case separately.

Consider the first case. Here, $\mathcal{I}(\mathbf{p}) = [-\sqrt{|\mathbf{p}|^2 + \Omega^2}, \sqrt{|\mathbf{p}|^2 + \Omega^2}]$. Recall that the poles of the integrand in equation \eqref{eq:cutoffG} occur at $p_0 = \pm\,\omega = \pm \sqrt{|\mathbf{p}|^2 + m^2}$, so if $m < \Omega$, the poles lie within the bounds of integration, and if $m > \Omega$, the poles lie outside the bounds. Only the first case is physical, as we cannot have masses whose magnitudes lie beyond the Planck scale, to which $\Omega$ is set. Thus, further assume that $m < \Omega$. We can then evaluate the integral in equation \eqref{eq:cutoffG} as follows:
\begin{align}
G_F^c(t = t^\prime, \mathbf{p}) &= \frac{i}{(2\pi)^{5/2}} \int_{-r_2}^{r_2} dp_0 \, \frac{1}{p_0^2 - \omega^2} \\
&= \frac{i}{(2\pi)^{5/2}}\left[ \int_{-\infty}^\infty dp_0 \, \frac{1}{p_0^2 - \omega^2} - \int_{r_2}^\infty dp_0 \, \frac{1}{p_0^2 - \omega^2} - \int_{-\infty}^{-r_2} dp_0 \, \frac{1}{p_0^2 - \omega^2} \right] \\
&=  \frac{1}{(2\pi)^{3/2}} \frac{1}{2\omega} - \frac{2i}{(2\pi)^{5/2}} \lim_{R \rightarrow \infty} \int_{r_2}^R dp_0 \, \frac{1}{p_0^2 - \omega^2}
\end{align}

\noindent where $r_2 \defeq \sqrt{|\mathbf{p}|^2 + \Omega^2}$. We have that
\begin{align}
\lim_{R \rightarrow \infty} \int_{r_2}^R dp_0 \, \frac{1}{p_0^2 - \omega^2} &= \lim_{R \rightarrow \infty} \left. \frac{1}{2\omega} \left( \ln | p_0 - \omega | - \ln |p_0 + \omega | \right) \right|_{r_2}^R \\
%&= \frac{1}{2\omega} \lim_{R \rightarrow \infty} \ln \left(\frac{R-\omega}{R+\omega} \right) + \frac{1}{2\omega} \ln \left( \frac{B+\omega}{B-\omega} \right) \\
&= \frac{1}{2\omega} \ln \left( \frac{r_2+\omega}{r_2-\omega} \right)
\end{align}

\noindent so for $|\mathbf{p} | \leq \Omega$ and $m < \Omega$,
\begin{equation}
G_F^c(t = t^\prime, \mathbf{p}) = \frac{1}{(2\pi)^{3/2}} \frac{1}{2\omega} - \frac{i}{(2\pi)^{5/2}} \frac{1}{\omega} \ln \left( \frac{r_2+\omega}{r_2-\omega} \right).
\end{equation}

For completeness, let us also calculate $G_F^c(t = t^\prime, \mathbf{p})$ for $m > \Omega$. In this case, the poles of the integrand lie outside the bounds of integration, so we can evaluate the integral using the fundamental theorem of calculus.
\begin{align}
G_F^c(t = t^\prime, \mathbf{p}) &= \frac{i}{(2\pi)^{5/2}} \int_{-r_2}^{r_2} dp_0 \, \frac{1}{p_0^2 - \omega^2} \\
&= \frac{i}{(2\pi)^{5/2}} \left. \frac{1}{2\omega} \left( \ln | p_0 - \omega | - \ln |p_0 + \omega | \right) \right|_{-r_2}^{r_2} \\
&= - \frac{i}{(2\pi)^{5/2}} \frac{1}{\omega} \ln \left( \frac{\omega+r_2}{\omega-r_2} \right)
\end{align}

Consider now the second case. Here, $\mathcal{I}(\mathbf{p}) = [-\sqrt{|\mathbf{p}|^2 + \Omega^2}, -\sqrt{|\mathbf{p}|^2 - \Omega^2} ] \cup [\sqrt{|\mathbf{p}|^2 - \Omega^2}, \sqrt{|\mathbf{p}|^2 + \Omega^2}]$. Again, the poles lie within the bounds of integration if $m < \Omega$ and outside the bounds if $m > \Omega$. Assuming $m < \Omega$ and proceeding similarly to before, we have
\begin{align}
G_F^c(t = t^\prime, \mathbf{p}) &= \frac{i}{(2\pi)^{5/2}} \left[ \int_{-r_2}^{-r_1} dp_0 \, \frac{1}{p_0^2 - \omega^2} + \int_{r_1}^{r_2} dp_0 \, \frac{1}{p_0^2 - \omega^2} \right] \\
&= \frac{i}{(2\pi)^{5/2}}\left[ \int_{-\infty}^\infty dp_0 \, \frac{1}{p_0^2 - \omega^2} \right. \\
\nonumber &\left. \qquad\qquad\qquad\qquad- \int_{-\infty}^{-r_2} dp_0 \, \frac{1}{p_0^2 - \omega^2} - \int_{-r_1}^{r_1} dp_0 \, \frac{1}{p_0^2 - \omega^2} - \int_{r_2}^{\infty} dp_0 \, \frac{1}{p_0^2 - \omega^2} \right] \\
&= \frac{1}{(2\pi)^{3/2}} \frac{1}{2\omega} - \frac{i}{(2\pi)^{5/2}} \frac{1}{\omega} \left[ \ln \left( \frac{r_2+\omega}{r_2-\omega} \right) - \ln \left( \frac{\omega+r_1}{\omega-r_1} \right) \right]
\end{align}

\noindent where $r_1 \defeq \sqrt{|\mathbf{p}|^2 - \Omega^2}$. In the case where $m > \Omega$,
\begin{align}
G_F^c(t = t^\prime, \mathbf{p}) &=  \frac{i}{(2\pi)^{5/2}} \left[ \int_{-r_2}^{-r_1} dp_0 \, \frac{1}{p_0^2 - \omega^2} + \int_{r_1}^{r_2} dp_0 \, \frac{1}{p_0^2 - \omega^2} \right] \\
&= \frac{2i}{(2\pi)^{5/2}} \int_{r_1}^{r_2} dp_0 \, \frac{1}{p_0^2 - \omega^2} \\
%&= \frac{2i}{(2\pi)^{5/2}} \left. \frac{1}{2\omega} \left( \ln | p_0 - \omega | - \ln |p_0 + \omega | \right) \right|_{b}^B \\
&= - \frac{i}{(2\pi)^{5/2}} \frac{1}{\omega} \left[ \ln \left( \frac{\omega+r_2}{\omega-r_2} \right) - \ln \left( \frac{\omega+r_1}{\omega-r_1} \right) \right]
\end{align}

\end{document}